\documentclass[11pt,a4paper,reqno]{amsart}
\usepackage{amsthm,amsmath,amsfonts,amssymb,amsxtra,appendix,bookmark,euscript,delarray,dsfont,mathrsfs}
\usepackage{enumerate}
\usepackage[latin1]{inputenc}

\theoremstyle{plain}
\newtheorem{theorem}{Theorem}
\newtheorem{lemma}[theorem]{Lemma}
\newtheorem*{lemma*}{Lemma}

\newtheorem*{conjecture*}{Conjecture}

\theoremstyle{definition}

\theoremstyle{remark}

\DeclareMathOperator{\supp}{supp}

\def\bq{\begin{eqnarray}}
\def\eq{\end{eqnarray}}
\def\bqq{\begin{eqnarray*}}
\def\eqq{\end{eqnarray*}}
\def\nn{\nonumber}

\def\epsilon{\varepsilon}

\newcommand{\norm}[1]{\left\lVert #1 \right\rVert}
\newcommand{\abs}[1]{\left\lvert#1\right\rvert}

\renewcommand{\phi}{\varphi}
\newcommand\1{{\ensuremath {\mathds 1} }}

\def\R{\mathbb{R}}

\def\wO{\widetilde{\Omega}}
\def\eps{\varepsilon}

\allowdisplaybreaks

\title[Lieb--Thirring inequality in the strong--coupling limit]{The Lieb--Thirring inequality for interacting systems in strong--coupling limit}

\author[K. K\"ogler]{Kevin K\"ogler}
\address{Department of Mathematics, LMU Munich, Theresienstrasse 39, 80333 Munich, Germany} 
\email{koegler.kevin@yahoo.com}

\author[P.T. Nam]{Phan Th\`anh Nam}
\address{Department of Mathematics, LMU Munich, Theresienstrasse 39, 80333 Munich, and Munich Center for Quantum Science and Technology (MCQST), Schellingstr. 4, 80799 Munich, Germany}
\email{nam@math.lmu.de}

\begin{document}

\begin{abstract}
	We consider an analogue of the Lieb--Thirring inequality for quantum systems with homogeneous repulsive interaction potentials, but without the antisymmetry assumption on the wave functions.
	We show that in the strong--coupling limit, the Lieb--Thirring constant converges to the optimal constant of the one-body Gagliardo--Nirenberg interpolation inequality without interaction. 
\end{abstract}
\maketitle

\section{Introduction}

The celebrated Lieb--Thirring inequality is a combination of the uncertainty and exclusion principles, two of the most important concepts in quantum mechanics. In the context of the kinetic energy of Fermi gases, it states that for any dimension $d\ge 1$,  the lower bound 
\begin{equation}\label{eq-LT}
\left\langle \Psi, \sum_{i=1}^N -\Delta_{x_i}   \Psi \right\rangle \geq \mathscr{K}_{\rm LT}(d) \int_{\R^d} \rho_\Psi^{1+\frac{2}{d}} (x) \,dx
\end{equation}
holds true for any wave functions $\Psi \in L^2(\R^{dN})$ that is normalized and anti-symmetric, namely
$\norm{\Psi}_{L^2(\R^{dN})}=1$ and 
\begin{equation} \label{eq-antisym}
\Psi(x_1,...,x_i,...,x_j,...,x_N)=-\Psi(x_1,...,x_j,...,x_i,...,x_N),\quad \forall i\neq j, \quad \forall x_i, x_j \in \R^d. 
\end{equation}
For any normalized wave function $\Psi \in L^2(\R^{dN})$, the function 
$$\rho_\Psi(x) := \sum_{j=1}^N \int_{\R^{(d-1)N}} \abs{\Psi(x_1,...,x_{j-1},x,x_{j+1},...,x_N)}^2 \prod_{i: i \neq j} \,dx_i$$
is called the one-body density of $\Psi$. We can interpret $\int_\Omega \rho_{\Psi}$ as the expected number of particles in $\Omega \subset \R^d$, in particular $\int_{\R^d}\rho_{\Psi}=N$ is the total number of particles. We ignore the spin of particles for simplicity. 

It is important that the constant $\mathscr{K}_{\rm LT}(d) >0$ in \eqref{eq-LT}  is independent of not only the wave function $\Psi$ but also the  particle number $N$. The inequality \eqref{eq-LT} was derived by Lieb and Thirring in 1975 as an essential tool in their proof of the stability of matter \cite{LiTh-75, LiTh-76}. Thanks to a standard duality argument, the kinetic bound \eqref{eq-LT} is equivalent to a lower bound on the sum of negative eigenvalues of Schr\"odinger operators $-\Delta +V(x)$ on $L^2(\R^d)$, making it very useful to semiclassical analysis (see \cite[Chapter 12]{LiLo-01} for a connection to Weyl's law).  In fact, up to a constant factor, the right side of \eqref{eq-LT} agrees with the Thomas--Fermi approximation for large $N$: 
$$
\left\langle \Psi, \sum_{i=1}^N -\Delta_{x_i}   \Psi \right\rangle \approx \mathscr{K}_{\rm cl} (d) \int_{\R^d} \rho_\Psi^{1+\frac{2}{d}}(x) \, dx 
$$
where 
$$
\mathscr{K}_{\rm cl}(d) := \frac{d}{d+2}\cdot \frac{4\pi^2}{|B(0,1)|^{2/d}}
$$
is called the semiclassical constant ($|B(0,1)|$ is the volume of the unit ball in $\R^d$). 

Note that the anti-symmetry condition \eqref{eq-antisym} is crucial for \eqref{eq-LT} to hold. Without Pauli's exclusion principle, the best bound one can get from the kinetic energy is 
\begin{equation}\label{eq-LT-bos}
\left\langle \Psi, \sum_{i=1}^N -\Delta_{x_i}   \Psi \right\rangle \geq \frac{C_{\rm GN}(d)}{N^{2/d}} \int_{\R^d} \rho_\Psi^{1+\frac{2}{d}} (x) \,dx
\end{equation}
where $C_{\rm GN}(d)$ is the sharp constant in the Gagliardo--Nirenberg interpolation inequality
\begin{equation}\label{eq-GN}
\Big( \int_{\R^d} |\nabla u(x)|^2 dx \Big) \Big( \int_{\R^d} |u(x)|^2 dx \Big)^{2/d} \ge  C_{\rm GN}(d) \int_{\R^d} |u(x)|^{2(1+2/d)} dx, \quad \forall u\in H^1(\R^d) 
\end{equation}
(see e.g. \cite{LuNaPo-16}). One can think of \eqref{eq-GN} as a quantitative version of the uncertainty principle. Clearly the lower bound \eqref{eq-LT-bos} is optimal when $N=1$, but it is not very useful when $N$ becomes large because the factor $N^{-2/d}$ on the right side becomes very small. The appearance of this  small factor is due to the fact that the particles are allowed to be stacked on top of each other, in which case the left side of  \eqref{eq-LT-bos} scales like $N$ while the integral on the right side scales like $N^{1+\frac{2}{d}}$. Intuitively, for an inequality similar to \eqref{eq-LT} to hold one needs some conditions to control such an overlapping of the particles, namely some version of the exclusion principle.

Computing the sharp constant $\mathscr{K}_{\rm LT}$ in the Lieb--Thirring inequality  \eqref{eq-LT} is an important open problem in mathematical physics. In \cite{LiTh-76}, Lieb and Thirring conjectured that 
$$
\mathscr{K}_{\rm LT}(d) = \min \{ \mathscr{K}_{\rm cl}(d) , C_{\rm GN} (d) \} = \begin{cases}
       \mathscr{K}_{\rm cl}(d) & \text{if $d\ge 3$,}\\
      C_{\rm GN} (d) & \text{if $d=1,2$.}
    \end{cases} 
$$
We refer to the recent work  \cite{FHJN-19} for the best known estimate on $\mathscr{K}_{\rm LT}(d)$. See also \cite{FGL-20} for a recent investigation on the conjectured  bounds on eigenvalues of Schr\"odinger operators. 

As regards to the semiclassical constant, it was proved in \cite{Nam-18}  that for all $d\ge 1$ and for all $\delta\in (0,1)$ one has
\begin{align} \label{eq:n-cl}
 \left\langle \Psi, \sum_{i=1}^N -\Delta_{x_i}   \Psi \right\rangle \ge (1-\delta)\mathscr{K}_{\rm cl}(d) \int_{\R^d} \rho_\Psi^{1+\frac{2}{d}} (x) \,dx - \frac{C_d}{\delta^{3+4/d}} \int_{\R^d} |\nabla \sqrt{\rho_\Psi}|^2
\end{align}
for any anti-symmetric normalized wave function $\Psi$ in $L^2(\R^{dN})$ and for any $N$. This bound implies the Lieb--Thirring inequality \eqref{eq-LT} with a non-sharp constant because the gradient term is bounded by the kinetic energy, thanks to the Hoffmann--Ostenhof inequality \cite{HO2-77}.  See \cite{LLS-20} for a related upper bound and the application in local density approximation, and see \cite{BVV18a,BVV18b} for discussions on related interpolation inequalities. 

In the present paper, we will focus on the one-body Gagliardo--Nirenberg constant and study its relation to a  Lieb--Thirring inequality with  repulsive interactions. In 2015, Lundholm, Portmann and Solovej \cite{LuPoSo-15} showed that for any dimension $d\ge 1$ and any constant $\lambda>0$, the Lieb--Thirring inequality 
\begin{equation} \label{eq-LT-potential-1}
\left\langle \Psi, \left(\sum_{i=1}^N -\Delta_{x_i}+ \sum_{1 \leq i<j \leq N} \frac{\lambda}{\abs{x_i-x_j}^2} \right)  \Psi \right\rangle 
\geq C_{\rm LT}(d,\lambda) \int_{\R^d} \rho_\Psi^{1+\frac{2}{d}} \,dx.
\end{equation}
holds true with a constant $C_{\rm LT}(d,\lambda)>0$ depending only on $d$ and $\lambda$. Remarkably, \eqref{eq-LT-potential-1}  holds true for any normalized wave function $\Psi$ in $L^2(\R^{dN})$, even without the anti-symmetry condition \eqref{eq-antisym}. Later, the bound \eqref{eq-LT-potential-1} was extended in \cite[Theorem 1]{LuNaPo-16} to the fractional case, namely for all $\lambda>0$ and $s>0$, one has the Lieb--Thirring inequality 
\begin{equation} \label{eq-LT-potential-s}
\left\langle \Psi, \left(\sum_{i=1}^N (-\Delta_{x_i})^s+ \sum_{1 \leq i<j \leq N} \frac{\lambda}{\abs{x_i-x_j}^{2s}} \right)  \Psi \right\rangle 
\geq C_{\rm LT}(s,d,\lambda) \int_{\R^d} \rho_\Psi^{1+\frac{2s}{d}} \,dx
\end{equation}
with a constant $C_{\rm LT}(s,d,\lambda)>0$. Here the power $2s$ in the interaction potential is the natural parameter such that the interaction energy and the kinetic energy scale the same under dilations.  In \cite{LuNaPo-16}, the authors also discussed briefly the behavior of the optimal constant $C_{\rm LT}(s,d,\lambda)$ in \eqref{eq-LT-potential-s}. They conjectured that in the strong--coupling limit $\lambda\to \infty$, the Lieb--Thirring constant  $C_{\rm LT}(s,d,\lambda)$ in \eqref{eq-LT-potential-s} converges to  the optimal  constant in the corresponding one-body Gagliardo--Nirenberg inequality. Heuristically, this conjecture is easy to understand because in the strong--coupling limit each particle is forced to stay away from the others and the many--body interacting system reduces to a one-body non-interacting system. However, proving this rigorously is nontrivial since we have to prove estimates uniformly in the number of particles. In the present paper, we will justify this
  conjecture rigorously. Moreover, in the case $2s<d$, we also obtain a similar result when the fractional Laplacian is replaced by the Hardy operator $(-\Delta)^{s} - \mathcal{C}_{s,d} |x|^{-2s}$. The precise statements of our results are presented in the next section.

To our knowledge,  there are very rare rigorous results connecting the Gagliardo--Nirenberg constant to Lieb--Thirring inequalities. In a remarkable paper in 2004 \cite{BL-00}, Benguria and Loss noticed that the Lieb--Thirring conjecture in the special case $d=1$ and $N=2$ is related to an open problem concerning a sharp  isoperimetric inequality for the lowest eigenvalue of a Schr\"odinger operator defined on a closed planar curve (see also \cite{BT-05,L-06} for related results). This is an illustration for the difficulty of the Lieb--Thirring conjecture when  the Gagliardo--Nirenberg constant is expected to emerge. We hope that our approach will give new insights to this challenging question.

\subsection*{Acknowledgment} We thank Simon Larson, Douglas Lundholm and Fabian Portmann for helpful discussions. We thank the referees for constructive comments and remarks. We received funding from the Deutsche Forschungsgemeinschaft (DFG, German Research Foundation) under Germany's Excellence Strategy (EXC-2111-390814868). 

\section{Main results} 

As usual, for any constant $s>0$ the operator $(-\Delta)^s$ on $L^2(\R^d)$ is defined via the Fourier transform
$$
  \widehat{(-\Delta)^s u}(p)= \abs{p}^{2s} \widehat{u}(p), \quad
  \widehat{u}(p):=\frac{1}{(2\pi)^{\frac{d}{2}}}\int_{\R^d}u(x)e^{-ip\cdot x}\,dx.
$$
The domain of $(-\Delta)^s$ is denoted by $H^s(\R^d)$ with the corresponding norm
$$
  \norm{u}_{H^s(\R^d)}^2:= \norm{u}_{L^2(\R^d)}^2+\norm{u}_{\dot{H}^s(\R^d)}^2 = \norm{u}_{L^2(\R^d)}^2 + \left\langle u,(-\Delta)^s u\right\rangle.
$$

\subsection{Lieb--Thirrring inequality with strong interactions}
Our first main result is 
\begin{theorem}[Lieb--Thirring constant in the strong--coupling limit] \label{thm-main1} Fix $d\ge 1$ and $s>0$. For any $\lambda>0$, let $C_{\rm LT}(s,d,\lambda)$ be the optimal constant in the  Lieb--Thirring inequality \eqref{eq-LT-potential-s}, namely
  \begin{align*}
	  C_{\rm LT}(s,d,\lambda) := \inf_{N\ge 2} \ \inf_{\substack{\Psi \in H^s(\R^{dN}) \\     
      \norm{\Psi}_{L^2}=1}} \frac{ \left\langle \Psi, \left( \sum\limits_{i=1}^N (-\Delta_{x_i})^s + \sum\limits_{1\le i<j\le N} \dfrac{ \lambda} { |x_i-x_j|^{2s}}  \right) \Psi \right\rangle}{ \int_{\R^d} \rho_{\Psi}^{1+\frac{2s}{d}} }.
  \end{align*}
Then we have
$$
\lim_{\lambda\to \infty} C_{\rm LT}(s,d,\lambda) = C_{\rm GN}(s,d)
$$
where 
\begin{align*}
	C_{\rm GN} (s,d):= \inf_{\substack{u \in H^s(\R^d)\\ 
			\norm{u}_{L^2} = 1}} \frac{\langle u, (-\Delta)^s u \rangle}{\int_{\R^d} |u|^{2 (1+\frac{2s}{d})}}.
  \end{align*}	
\end{theorem}	

Remarks:

\begin{itemize}

\item[1.] For the general power $s>0$, the Lieb--Thirring inequality for fermions
\begin{equation}\label{eq-LT-s}
\left\langle \Psi, \sum_{i=1}^N (-\Delta_{x_i})^s   \Psi \right\rangle \geq \mathscr{K}_{\rm LT}(s,d) \int_{\R^d} \rho_\Psi^{1+\frac{2s}{d}} 
\end{equation}
with $\mathscr{K}_{\rm LT}(s,d)>0$ was proved by Daubechies in 1983 \cite{Daubechies-83}. The optimal constant $\mathscr{K}_{\rm LT}(s,d)$ is unknown; see  \cite[Theorem 2]{FHJN-19} for a recent estimate. Without the anti-symmetry condition \eqref{eq-antisym}, the best replacement for \eqref{eq-LT-s} is
\begin{equation}\label{eq-LT-bos-s}
\left\langle \Psi, \sum_{i=1}^N (-\Delta_{x_i})^s   \Psi \right\rangle \geq \frac{C_{\rm GN}(s,d)}{N^{2s/d}} \int_{\R^d} \rho_\Psi^{1+\frac{2s}{d}} 
\end{equation}
with the Gagliardo--Nirenberg constant  $C_{\rm GN}(s,d)$ given in Theorem \ref{thm-main1} (see e.g. \cite{LuNaPo-16}). In Theorem \ref{thm-main1}, we do not assume the  anti-symmetry condition \eqref{eq-antisym} but we save the factor $N^{2s/d}$ on the right side of \eqref{eq-LT-bos-s} provided that the interacting term is  sufficiently strong. 

\item[2.] The upper bound $C_{\rm LT}(s,d,\lambda)\le C_{\rm GN}(s,d)$ can be seen easily by putting each of the $N$ particles far from the  others; see \cite[Proposition 10]{LuNaPo-16} for details. The difficult direction of Theorem \ref{thm-main1} is the lower bound. Here in the definition of $C_{\rm LT}(s,d,\lambda)$  in Theorem \ref{thm-main1} we put $\inf_{N\ge 2}$, but the result remains the same if we put  $\inf_{N\ge m}$ for any $m\ge 1$ (for $m=1$ we recover $C_{\rm GN}(s,d)$ trivially as there is no interaction term).  In general, if we introduce the constant $C_{\rm LT} (s,d,\lambda,N)$, with obvious definition, then it is decreasing in $N$, and hence the infimum is always attained in the limit $N\to \infty$. 

\item[3.] Theorem \ref{thm-main1} establishes a conjecture in \cite{LuNaPo-16}. As also conjectured in \cite{LuNaPo-16} and proved recently in \cite{LaLuNa-19},  $\lim_{\lambda\to 0^+} C_{\rm LT}(s,d,\lambda)$ is nontrivial (i.e. strictly positive) if and only if $2s>d$.  Theorem \ref{thm-main1} thus completes the picture on the range of the Lieb--Thirring constant $C_{\rm LT}(s,d,\lambda)$ with $\lambda\in (0,\infty)$.

\item[4.] When $d=1$ and $s=1$, the interaction potential is so singular that the wave functions in its quadratic form domain must vanish on the diagonal set \cite{LuNaPo-16,LaLuNa-19}. Therefore, by the well-known bosonic-fermionic correspondence in one dimension \cite{Girardeau-60}, we obtain  $\lim_{\lambda\to 0^+} C_{\rm LT}(1,1,\lambda) = \mathscr{K}_{\rm LT}(1)$, the fermionic Lieb--Thirring constant in \eqref{eq-LT}. Thus given Theorem \ref{thm-main1}, in order to prove the Lieb-Thirring conjecture $\mathscr{K}_{\rm LT}(1)= C_{\rm GN}(1)$, it remains to show that the constant $C_{\rm LT}(1,1,\lambda)$ is independent of $\lambda$. Note that $C_{\rm LT}(s,d,\lambda)$ is always increasing in $\lambda$. 

\item[5.] One may ask about other types of interactions which would be sufficient for exclusion (e.g. a nearest-neighbor type interaction) as well as the $\lambda$ dependence of the constant $C_{\rm LT}(s,d,\lambda)$. In principle our method is constructive and could be adapted to address these issues, but we do not pursue these  directions. 
\end{itemize}

\subsection{Hardy--Lieb--Thirring inequality  with strong interactions} Now we focus on the case $0<2s<d$, where we have the Hardy inequality 
$$
(-\Delta)^{s} - \mathcal{C}_{s,d} |x|^{-2s} \ge 0\quad \text{on }L^2(\R^d)
$$
with the sharp constant
$$
	\mathcal{C}_{s,d} := 2^{2s}\left( \frac{\Gamma((d+2s)/4)}{\Gamma((d-2s)/4)} \right)^2.
$$
The following improvement of \eqref{eq-LT-potential-s} has been proved in \cite[Theorem 2]{LuNaPo-16}
\begin{align}\label{eq:HLT_frac}
		\left\langle \Psi, \left(  \sum_{i=1}^N \left( (-\Delta_{i})^s - \frac{\mathcal{C}_{s,d}}{|x_i|^{2s}} \right) 
		+ \sum_{1\le i<j\le N} \frac{\lambda}{|x_i-x_j|^{2s}}\right) \Psi \right\rangle \ge 
		 C_{\rm HLT}(s,d,\lambda) \int_{\R^d} \rho_\Psi^{1+\frac{2s}{d}}. 
	\end{align}
	This inequality holds for any normalized wave function $\Psi\in L^2(\R^{dN})$ (without the anti-symmetry condition) and the constant $C_{\rm HLT}(s,d,\lambda)>0$ is independent of $N$ and $\Psi$. For $s=1/2$ and $d=3$, the left side of \eqref{eq:HLT_frac} can be interpreted 
as the energy of a system of $N$ relativistic quantum electrons moving around a classical nucleus fixed at the origin and interacting via Coulomb forces.

As explained in \cite[Theorem 4]{LuNaPo-16}, the  Hardy--Lieb--Thirring inequality \eqref{eq:HLT_frac} with a non-sharp constant is equivalent to the one-body interpolation inequality
	\begin{align} \label{eq:HLT-inter-1}
		\left\langle u, \Big((-\Delta)^s - \mathcal{C}_{s,d} |x|^{-2s} \Big)  u \right\rangle^{1-2s/d} 
		& \left( \iint_{\R^d \times \R^d} \frac{|u(x)|^2|u(y)|^2}{|x-y|^{2s}}\,dxdy \right)^{2s/d}  \\
		&\ge C(s,d) \int_{\R^d} |u(x)|^{2(1+2s/d)}\,dx, \quad \forall u\in H^s(\R^d). \nn
	\end{align}
A slightly weaker version of \eqref{eq:HLT-inter-1}, when the Hardy potential $- \mathcal{C}_{s,d} |x|^{-2s}$ is removed, has been proved by Bellazzini, Ozawa and Visciglia for the case $s=1/2,d=3$  \cite{BelOzaVis-11}, and by Bellazzini, Frank and Visciglia for the general case $0<s<d/2$ \cite{BelFraVis-14}. 

In the present paper we consider the asymptotic behavior of the optimal constant $C_{\rm HLT}(s,d,\lambda)$ in \eqref{eq:HLT_frac} when  $\lambda\to \infty$. Similarly to 
Theorem  \ref{thm-main1}, we have

\begin{theorem}[Hardy--Lieb--Thirring constant in the strong--coupling limit] \label{thm-main2}  Fix $0<2s<d$.  For any $\lambda>0$, let $C_{\rm HLT}(\lambda)$ be the optimal constant in the Hardy--Lieb--Thirring inequality \eqref{eq:HLT_frac}, namely
$$
C_{\rm HLT}(s,d,\lambda):=\inf_{N\ge 2}\ \inf_{\substack{ \Psi \in H^s(\R^{dN}) \\     
      \norm{\Psi}_{L^2}=1}} \frac{ \left\langle \Psi, \left( \sum\limits_{i=1}^N \Big( (-\Delta_{x_i})^s - \dfrac{\mathcal{C}_{s,d}}{|x_i|^{2s}} \Big) +  \sum\limits_{1\le i<j\le N} \dfrac{\lambda}{|x_i-x_j|^{2s}} \right) \Psi \right\rangle}{ \int_{\R^d} \rho_{\Psi}^{1+\frac{2s}{d}} }.
$$
Then we have
$$
\lim_{\lambda\to \infty} C_{\rm HLT}(s,d,\lambda) = C_{\rm HGN}(s,d)
$$
where 
\begin{align*}
	C_{\rm HGN}(s,d) := \inf_{\substack{u \in H^s(\R^d)\\ 
			\norm{u}_{L^2} = 1}} \frac{\Big\langle u, \Big((-\Delta)^s - \mathcal{C}_{s,d} |x|^{-2s} \Big) u \Big\rangle}{\int_{\R^d} |u|^{2 (1+\frac{2s}{d})}}.
  \end{align*}	
\end{theorem}

Remarks:

\begin{itemize}

\item[1.] The Hardy--Lieb--Thirring inequality for fermions
\begin{equation}\label{eq-HLT-s}
\left\langle \Psi, \sum_{i=1}^N \Big( (-\Delta_{x_i})^s - \mathcal{C}_{s,d} |x|^{-2s}  \Big)   \Psi \right\rangle \geq \mathscr{K}_{\rm HLT}(s,d) \int_{\R^d} \rho_\Psi^{1+\frac{2s}{d}} 
\end{equation}
with $\mathscr{K}_{\rm HLT}(s,d)>0$ was proved for the non-relativistic case $s=1$ by  Ekholm and Frank in 2006
\cite{EkhFra-06}, for the fractional powers $0<s\le 1$ by Frank, Lieb and Seiringer  in 2008 \cite{FraLieSei-08}, and for the full range $0<s<d/2$ by Frank in 2009 \cite{Frank-09}. The sharp constant $\mathscr{K}_{\rm HLT}(s)>0$ is unknown. 

\item[2.] The upper bound $C_{\rm HLT}(s,d,\lambda)\le C_{\rm HGN}(s,d)$ is easy to see by putting one particle close to the origin and putting $N-1$ particles at infinity such that each particle is far from the others.  The main point of Theorem \ref{thm-main2} is the lower bound. 

\item [3.] From our proof, it is possible to extract an explicit error estimate for the convergence $\lim_{\lambda\to \infty} C_{\rm HLT}(s,d,\lambda) = C_{\rm HGN}(s,d)$ in Theorem \ref{thm-main2}  (as well as the convergence $\lim_{\lambda\to \infty} C_{\rm LT}(s,d,\lambda) = C_{\rm GN}(s,d)$ in Theorem \ref{thm-main1}) in terms of $\lambda$. We will not do it here in order to keep the proof ideas more transparent. 
\end{itemize}

\subsection{Proof strategy} We will use the method of microlocal analysis. The idea goes back to the seminal work of Dyson and Lenard in 1967 \cite{DyLe-67,DyLe-68} where they proved the stability of matter using only a local formulation of the exclusion principle which is a relatively weak consequence of \eqref{eq-antisym}. In 2013, Lundholm and Solovej \cite{LuSo-13} found that one can actually obtain the Lieb-Thirring inequality \eqref{eq-LT} (with a non-sharp constant) by combining the local exclusion in \cite{DyLe-67,DyLe-68} with a local formulation of the uncertainty principle. They used this method to derive the Lieb--Thirring inequality for particles with fractional
statistics  in one and two dimensions \cite{LuSo-13,LuSo-13b,LuSo-14}.  Later, this method has been developed by many authors to derive several new Lieb--Thirring-type inequalities  \cite{FrSe-12,LuPoSo-15,LuNaPo-16,Nam-18,LaLu-18,LuSe-18,LaLuNa-19}. We refer to Lundholm's lecture notes \cite{Lun-18} for a pedagogical discussion.  All of the existing results are not concerned with the optimal constants, except the fermionic semiclassical bound \eqref{eq:n-cl}  in \cite{Nam-18}.  

In the present paper, we will revisit and improve the microlocal analysis for interacting systems developed in \cite{LuSo-13,LuPoSo-15,LuNaPo-16,LaLuNa-19}. We follow the overall strategy in \cite{LuNaPo-16}, by combining some local uncertainty and exclusion on an appropriate covering of the support of $\rho_{\Psi}$. In order to recover the sharp Gagliardo--Nirenberg constant in the strong--coupling limit, we need three new ingredients. 

\begin{itemize}

\item For an open bounded set $\Omega\subset \R^d$, consider the localized kinetic operator $(-\Delta)^s_{|\Omega}$ on $L^2(\R^d)$ defined by  
$$
\langle u, (-\Delta)^s_{|\Omega} u\rangle_{L^2(\R^d)} = \|u \|^2_{\dot H^s(\Omega)} 
$$
(see Section \ref{sec4} for details). The uncertainly principle in \cite[Lemma 8]{LuNaPo-16} states that 
\begin{equation} \label{eq-local-uncertainty-I-intro}
	\Big\langle \Psi, \sum_{i=1}^N (-\Delta_{x_i})^{s}_{|\Omega} \Psi \Big\rangle \ge \frac{1}{C(s,d)} \frac{\int_\Omega \rho_\Psi^{1+2s/d} }{\Big( \int_\Omega \rho_\Psi \Big)^{2s/d}}  - \frac{C(s,d)}{|\Omega|^{2s/d}} \int_\Omega \rho_\Psi. 
\end{equation}
When the mass $\int_\Omega \rho_\Psi$ is not small, the desired term $\int_{\Omega} \rho_\Psi^{1+2s/d}$ in \eqref{eq-local-uncertainty-I-intro}  will be coupled with a small factor. In the present paper, we improve this  by  making the optimal constant $C_{\rm GN}(s,d)$ appear explicitly. Roughly speaking, in  Lemma \ref{lem:LUP-II} we prove  that for bounded domains $\Omega \subset \subset \wO \subset \R^d$,  which up to translation and dilation belong to a finite collection of sets,  and for any $\delta>0$ small, 
\begin{equation} \label{eq-local-uncertainty-II-intro}
	\Big\langle \Psi, \sum_{i=1}^N (-\Delta_{x_i})^{s}_{|\wO} \Psi \Big\rangle \ge C_{\rm GN}(s,d)  (1-\delta)\frac{\int_\Omega \rho_\Psi^{1+2s/d} }{\Big( \int_{\wO} \rho_\Psi \Big)^{2s/d}}  - \frac{C(s,d,\delta)}{|\Omega|^{2s/d}} \int_{\wO}\rho_\Psi. 
	\end{equation}
	This bound  is useful when the mass $\int_{\wO} \rho_\Psi$ is smaller than $1+\delta$. 
	
\item In order to control the error in the local uncertainty principle, namely the last term of \eqref{eq-local-uncertainty-II-intro}, we need to use the interaction energy. A lower bound for the interaction energy in cubes is given in \cite[Lemma 6]{LuNaPo-16} (see also \cite[Theorem 2]{LuPoSo-15}). In Lemma \ref{lem-interaction-K2}, we prove a refined version of the local exclusion principle which allows the flexibility of the diameter of the sets; i.e. the bound is good for not only cubes, but also for ``clusters of cubes". More precisely, we prove that if $\{\Omega_{n,m}\}_{n,m\ge 1}$ is a collection of sets in $\R^d$ such that 
$$
{\rm diam}(\Omega_{n,m}) \le \eps^n, \quad \int_{\Omega_{n,m}} \rho_\Psi \ge 1+\delta, \quad \sum_{m} \1_{\Omega_{n,m}} \le M
$$
with fixed parameters $\eps,\delta \in (0,1)$ and $M>0$, then 
 \begin{equation} \label{eq-local-exclusion-intro}
	\Big\langle \Psi, \sum_{1\le i<j \le N} \frac{1}{|x_i-x_j|^{2s}} \Psi \Big\rangle \ge \sum_{n\ge 1} \frac{1}{C(s,\eps,\delta,M) \eps^{2sn}} \sum_{m\ge 1} \int_{\Omega_{n,m}}  \rho_\Psi. 
	\end{equation}
	Heuristically, it is clear that we can extract a nontrivial contribution from the interaction energy in a set with a small diameter (i.e. ${\rm diam}(\Omega_{n,m}) \le \eps^n$) if the local mass is large enough (i.e. $\int_{\Omega_{n,m}} \rho_\Psi \ge 1+\delta$). The significance of \eqref{eq-local-exclusion-intro} is that we can count the interaction contribution from {\em all} sets of different length scales,  provided that the sets of each length scale do not overlap too much (i.e. $\sum_{m} \1_{\Omega_{n,m}} \le M$). In this way, we allow a huge overlap from the sets of different length scales (i.e. $\sum_{n,m} \1_{\Omega_{n,m}}$ can be arbitrarily large), which is important in application. The smallness factor of $1/C(s,\eps,\delta,M)$ will be compensated by the large coupling constant $\lambda$ in the interaction potential.  
		 
\item Most importantly,  we introduce a new construction of covering sub-cubes for the support of $\rho_\Psi$, which is  very flexible and hopefully will be useful in other contexts. In  \cite{LuNaPo-16}, the support of $\rho_{\Psi}$ is covered by disjoint cubes which are obtained by a standard {\em stopping time argument}: any cube $Q$ with the mass $\int_Q \rho_\Psi$ bigger than a given quantity  will be divided into $2^d$ sub-cubes. In this way, the masses in final sub-cubes may differ up to a factor $2^d$, leading to a similar factor loss in the Lieb--Thirring constant. In the present paper, in order to have access to the optimal constant $C_{\rm GN}(s,d)$, we apply the stopping time argument to ``clusters of cubes" rather than to individual cubes. More precisely, by induction, in the $n$-th step we obtain a collection $G^{n}$ of cubes of side length $\eps^{n}$ with a fixed parameter $\eps \in (0,1)$, which can be decomposed further into three disjoint sub-collections 
  $$G^{n}=G^{n,0}\bigcup G^{n,1}\bigcup G^{n,2}.$$
Heuristically, $G^{n,0} \bigcup G^{n,1}$ contains ``good sets" concerning the uncertainty principle. More precisely, $G^{n,0}$ contains the cubes such that the mass in each cube is less than $\delta$, making \eqref{eq-local-uncertainty-I-intro} useful. Moreover, $G^{n,1}$ contains the cubes that can be distributed to ``disjoint clusters" such that the mass in each cluster is smaller than $1+\delta$, making \eqref{eq-local-uncertainty-II-intro} useful. Technically, thanks to the removal of the  cubes in $G^{n,0}$, all clusters in $G^{n,1}$, up to translation and dilation, must belong to a finite collections of sets which is important to apply Lemma \ref{lem:LUP-II}. On the other hand, $G^{n,2}$ contains disjoint clusters such that the mass in each cluster is bigger than $1+\delta$, making the exclusion principle in \eqref{eq-local-exclusion-intro} useful. Finally, in order to estimate the kinetic energy of the cubes in $G^{n,0}$, we have to divide them further and obtain a collection $G^{n+1}$ of cubes of side length $\eps^{n+1}$. Since the cubes in $G^{n,0}$ cover the cubes in $G^{n+1,0}\bigcup G^{n+1,1}$, the interaction energy from $G^{n,0}$ can be used to compensate for the error resulted from applying the uncertainty principle to $G^{n+1,0}\bigcup G^{n+1,1}$.

%

\end{itemize}

The paper is structured as follows. We discuss the local uncertainty principle in Section \ref{sec4}  and the local exclusion principle in Section \ref{sec5}. Then in Section  \ref{sec6} we explain the construction of covering sub-cubes and prove Theorem \ref{thm-main1}. In Section  \ref{sec7} we provide the proof of Theorem \ref{thm-main2}, which follows the same overall approach of Theorem \ref{thm-main1} but the detailed analysis is more complicated because we have to deal with the singularity of the negative external potential. 

In the rest of the paper we will denote by $C$ a general large constant whose value may change from line to line. In some cases, the dependence on a given parameter will be noted, e.g. $C_\delta$ depends on $\delta$. We will often ignore the dependence of the dimension $d$ and the power $s$ to simplify the notation (e.g. we will simply write $C_{\rm LT}(\lambda)$ and $C_{\rm GN}$ for the constants $C_{\rm LT}(s,d,\lambda)$ and $C_{\rm GN}(s,d)$  in Theorem \ref{thm-main1}). On the other hand, it is important that all constants are always independent of the wave function $\Psi$ and the number of particles $N$.

\section{Local uncertainty} \label{sec4}

In this section we discuss Gagliardo--Nirenberg inequalities on bounded domains.

Let us recall the definition of the fractional Sobolev space; classical references are \cite{Ad-75,Tr-78}. For any power $s>0$, we write $s=m +\sigma$ with $m\in \mathbb{N}$ and $0 \le \sigma <1$. Then by working in the Fourier space it is straightforward to check that 
$$
\left\langle u,(-\Delta)^s u\right\rangle
=\sum_{\abs{\alpha}=m}\frac{m!}{\alpha!} \left\langle D^\alpha u,(-\Delta)^\sigma D^\alpha u\right\rangle.
$$
Here for $\alpha=(\alpha_1,...,\alpha_d)\in \{0,1,2,...\}^d$ and $x=(x^{(1)},...,x^{(d)}) \in \R^d$ we denoted 
$$
\alpha! = \alpha_1! \alpha_2! ... \alpha_d!  , \quad D^\alpha= \partial_{x^{(1)}}^{\alpha_1} \partial_{x^{(2)}}^{\alpha_2}  ... \partial_{x^{(d)}}^{\alpha_d}.
$$
If $\sigma=0$, then $(-\Delta)^{\sigma}=\1$ (the identity). Moreover, for  $0 < \sigma <1$ we have the well-known identity (see e.g. \cite[Lemma 3.1]{FraLieSei-08}) 
$$\left\langle u,(-\Delta)^\sigma u\right\rangle= c_{d,\sigma} \int_{\R^d}\int_{\R^d} \frac{\abs{u(x)-u(y)}^2}{\abs{x-y}^{d+2\sigma}}\,dxdy,\quad c_{d,\sigma}:=\frac{2^{2\sigma-1}}{\pi^{d/2}} \frac{\Gamma(d/2+\sigma)}{|\Gamma(-\sigma)|},$$
and hence 
$$\left\langle u,(-\Delta)^s u\right\rangle
= c_{d,\sigma} \sum_{\abs{\alpha}=m}\frac{m!}{\alpha!} \int_{\R^d}\int_{\R^d} \frac{\abs{D^\alpha u(x)-D^\alpha u(y)}^2}{\abs{x-y}^{d+2\sigma}}\,dxdy$$
For a domain $\Omega \subset \R^d$ we introduce the seminorm $\| \cdot\|_{\dot H^s(\Omega)}$ by
$$\norm{u}_{\dot{H}^s(\Omega)}^2=\sum_{\abs{\alpha}=m}\frac{m!}{\alpha!} \int_\Omega \abs{D^\alpha u}^2\,dx, \quad \text{ if $\sigma=0$ (i.e. $s=m$)}$$
and
$$\norm{u}_{\dot{H}^s(\Omega)}^2= c_{d,\sigma}\sum_{\abs{\alpha}=m}\frac{m!}{\alpha!}  \int_{\Omega}\int_{\Omega} \frac{\abs{D^\alpha u(x)-D^\alpha u(y)}^2}{\abs{x-y}^{d+2\sigma}}\,dxdy, \quad \text{ if $0< \sigma <1$}.$$
We define the operator $(-\Delta)^s_{|\Omega}$ on $L^2(\R^d)$ via the quadratic form formula
$$
\langle u, (-\Delta)^s_{|\Omega} u\rangle_{L^2(\R^d)} = \|u \|^2_{\dot H^s(\Omega)}, \quad \forall u\in H^s(\R^d) 
$$
and Friedrichs' extension. Note that $\|u \|^2_{\dot H^s(\Omega)}$ depends only on $u_{|\Omega}$, and hence we can also restrict $(-\Delta)^s_{|\Omega}$ to $L^2(\Omega)$ using the same quadratic form formula. (The reason we want to think of $(-\Delta)^s_{|\Omega}$ as an operator on $L^2(\R^d)$ is that later we can write $\langle \Psi, (-\Delta_{x_i})^s_{|\Omega} \Psi\rangle$ for wave functions $\Psi$ in $L^2(\R^{dN})$.) We denote by  $H^s(\Omega)$ the space of all functions $u: \Omega \to \mathbb{C}$ such that the norm
$$\norm{u}_{H^s(\Omega)}^2:= \norm{u}_{\dot{H}^s(\Omega)}^2 + \sum_{\abs{\alpha}\le m}\int_\Omega \abs{D^\alpha u}^2\,dx$$
is finite. Note that for disjoint domains $\{\Omega_i\}_i$ of $\R^d$ we have  the monotonicity 
$$\sum_i \norm{u}_{\dot{H}^s(\Omega_i)}^2 \leq  \norm{u}_{\dot{H}^s({\bigcup}\Omega_i)}^2.$$

Recall the Gagliardo--Nirenberg inequality \eqref{eq-LT-bos-s}: for any normalized wave function $\Psi\in L^2(\R^{dN})$ we have
$$
\left\langle \Psi, \sum_{i=1}^N (-\Delta_{x_i})^s   \Psi \right\rangle \geq \frac{C_{\rm GN}}{N^{2s/d}} \int_{\R^d} \rho_\Psi^{1+\frac{2s}{d}} 
$$
where $C_{\rm GN}$ is the optimal constant in the one-body case
\begin{align*}
	C_{\rm GN} := \inf_{\substack{u \in H^s(\R^d)\\ 
			\norm{u}_{L^2} = 1}} \frac{\langle u, (-\Delta)^s u \rangle}{\int_{\R^d} |u|^{2 (1+\frac{2s}{d})}}.
  \end{align*}
The bound \eqref{eq-LT-bos-s} is not very useful when $N$ becomes large. However, we can derive its local versions which are more powerful. Let us recall a key estimate from \cite[Proof of Lemma 8]{LuNaPo-16}. 

\begin{lemma}[Local uncertainty principle I] \label{lem:LUP} Let $d \geq 1$, $s>0$. Let $\Psi$ be a normalized wave function in $L^2(\R^{dN})$. Then for any cube $Q \subset \R^d$ we have
	$$
	\Big\langle \Psi, \sum_{i=1}^N (-\Delta_{x_i})^{s}_{|Q} \Psi \Big\rangle \ge \frac{1}{C} \frac{\int_Q \rho_\Psi^{1+2s/d} }{\Big( \int_Q \rho_\Psi \Big)^{2s/d}}  - \frac{C}{|Q|^{2s/d}} \int_Q \rho_\Psi. 
	$$
	Here the constant $C=C(d,s)>0$ is independent of $Q, \Psi$ and $N$. 
\end{lemma}	

In  \cite{LuNaPo-16}, the bound in Lemma \ref{lem:LUP} is used for the cubes $Q$ such that $\int_Q \rho_\Psi$ is bounded independently of $Q$ and $N$, leading to estimates uniform in $N$.   In the present paper, we need a refined version of Lemma \ref{lem:LUP} which gives access to the optimal Gagliardo-Nirenberg constant $C_{\rm GN}$. 

We will state our results here for a general {\em $s$-extension domain} $\Omega$, namely an open subset of $\R^d$ such that there exists a linear operator $T$ mapping functions defined a.e. in $\Omega$ to functions defined a.e. in $\R^d$ satisfying that for all $0\le t \le s$,  
$$T: H^t(\Omega)\to H^t(\R^d) \text{ is a bounded linear operator}, \quad Tu_{|\Omega}=u, \quad \forall u \in H^t(\Omega).$$
For our application, a cube or a finite union of  connected cubes is a $s$-extension domain for all $s>0$ (see e.g.  \cite[Theorem 7.41]{Ad-75} or  \cite[Theorem 4.2.3]{Tr-78}). Our new result is

\begin{lemma}[Local uncertainty principle II] \label{lem:LUP-II} Let $d \geq 1$, $s>0$. Consider two domains $\Omega \subset\subset \widetilde \Omega \subset \R^d$ where $\wO$ is a $s$-extension domain. Then for any normalized wave function $\Psi$  in $L^2(\R^{dN})$  and any constant $\delta\in (0,1)$ we have
	$$
	\Big\langle \Psi, \sum_{i=1}^N (-\Delta_{x_i})^{s}_{|\wO} \Psi \Big\rangle \ge C_{\rm GN} (1-\delta)\frac{\int_\Omega \rho_\Psi^{1+2s/d} }{\Big( \int_{\wO} \rho_\Psi \Big)^{2s/d}}  - C_{\delta,\Omega,\wO} \int_{\wO}\rho_\Psi. 
	$$
	The constant $C_{\delta,\Omega,\wO}>0$ is independent of $\Psi$ and $N$. Moreover, it scales as 
	$$
	C_{\delta,\Omega,\wO}= C_{\delta,L\Omega,L\wO} L^{2s}, \quad \forall L>0. 
	$$
\end{lemma}	

Here we write $\Omega\subset \subset \wO$  when  $\Omega\subset \wO$ and ${\rm dist}(\Omega, \R^d\backslash \wO)>0$. As usual $L\Omega:=\{Lx,x\in \Omega\}$. The scaling property  $C_{\delta,\Omega,\wO}=C_{\delta,L\Omega,L\wO} L^{2s}$ follows by a change of variables. 

We will deduce Lemma \ref{lem:LUP-II} from its one-body version 
	\begin{equation} \label{eq-local-uncertainty-II-one-body}
	\|u \|^2_{\dot H^s(\wO)} \ge C_{\rm GN} (1-\delta)\frac{\int_\Omega |u|^{2(1+2s/d)} }{\Big( \int_{\wO} |u|^2 \Big)^{2s/d}}  - C_{\delta,\Omega,\wO} \int_{\wO} |u|^2, \quad \forall u \in H^s(\R^d). 
	\end{equation}

\begin{proof}[Proof of Lemma \ref{lem:LUP-II} using \eqref{eq-local-uncertainty-II-one-body}] We follow the proof strategy in  \cite[Lemma 8]{LuNaPo-16}. We introduce the one-body density matrix $\gamma_\Psi: L^2(\R^d)\to L^2(\R^d)$ given by the kernel 
$$\gamma_\Psi(x,y):=  \sum_{j=1}^N \int_{\R^{(d-1)N}} \Psi(x_1,...,x_{j-1},x,x_{j+1},...,x_N) \overline{\Psi(x_1,...,x_{j-1},y,x_{j+1},...,x_N)} \prod_{i \neq j} \,dx_i.$$
Since $\gamma_\Psi$ is a non-negative trace class operator on $L^2(\R^d)$, we can write 
$$\gamma_\Psi(x,y)=\sum_{n \geq 1}u_n(x)\overline{u_n(y)}$$
with an orthogonal family $\{u_n\}_{n\ge 1} \subset L^2(\R^d)$ (the functions $u_n$'s are not necessarily normalized in $L^2(\R^d)$).   
From this representation one obtains $\rho_\Psi= \sum_{n\geq 1} \abs{u_n}^2$ and
\begin{align*}
\left\langle \Psi, \sum_{i=1}^N (-\Delta_{x_i})^s_{|\wO} \Psi \right\rangle
=  \text{Tr} \left[ (-\Delta)^s_{|\wO} \gamma_\Psi\right] 
=\sum_{n \geq 1} \left\langle u_n, (-\Delta)^s _{|\wO}u_n \right\rangle = \sum_{n \geq 1} \|u_n \|^2_{\dot H^s(\wO)}. 
\end{align*} 
Using the triangle inequality, the one-body bound \eqref{eq-local-uncertainty-II-one-body} and H\"older inequality for sums we can bound 
\begin{align*}
&\Big[ C_{\rm GN} (1-\delta) \Big]^{1/(1+2s/d)} \|\rho_\Psi\|_{L^{1+2s/d}(\Omega)} = \Big[ C_{\rm GN} (1-\delta) \Big]^{1/(1+2s/d)} \|\sum_{n\ge 1} |u_n|^2\|_{L^{1+2s/d}(\Omega)}\\
&\le \sum_{n\ge 1} \Big[ C_{\rm GN} (1-\delta) \Big]^{1/(1+2s/d)} \||u_n|^2\|_{L^{1+2s/d}(\Omega)}\\
&\le \sum_{n\ge 1} \Big[  \|u_n\|^2_{\dot H^s(\wO)}  + C_{\delta,\Omega,\wO}   \int_{\wO} |u_n|^2  \Big]^{d/(d+2s)} \Big[ \int_{\wO} |u_n|^2 \Big]^{2s/ (d+2s)}\\
&\le \Big[  \sum_{n\ge 1} \|u_n\|^2_{\dot H^s(\wO)}  + C_{\delta,\Omega,\wO}   \sum_{n\ge 1} \int_{\wO} |u_n|^2  \Big]^{d/(d+2s)} \Big[  \sum_{n\ge 1} \int_{\wO} |u_n|^2 \Big]^{2s/ (d+2s)} \\
& = \Big[  \Big\langle \Psi, \sum_{i=1}^N (-\Delta_{x_i})^s_{|\wO} \Psi \Big\rangle  + C_{\delta,\Omega,\wO}   \int_{\wO}\rho_\Psi   \Big]^{d/(d+2s)} \Big[   \int_{\wO} \rho_\Psi  \Big]^{2s/ (d+2s)}.
\end{align*}
This is equivalent to the desired inequality  in Lemma \ref{lem:LUP-II}.
\end{proof}

It remains to prove \eqref{eq-local-uncertainty-II-one-body}. We will need the following general estimates for Sobolev norms on extension domains.

\begin{lemma}[Comparison of Sobolev norms] \label{lem-norm-equivalence} Let $d\ge 1$ and $s>0$. For any $s$-extension domain $\Omega \subset \R^d$ and $u \in H^s(\Omega)$ we have
	\begin{equation} \label{eq-norm-equivalence}
	\norm{u}^2_{H^s(\Omega)} \leq C(\norm{u}^2_{\dot{H}^s(\Omega)} + \norm{u}^2_{L^2(\Omega)}).
	\end{equation}
	Moreover, for any $t\in (0,s)$ and $\delta>0$ we have
	\begin{equation} \label{eq-norm-equivalence-II}
	\norm{u}^2_{H^t(\Omega)} \leq \delta \norm{u}^2_{\dot{H}^s(\Omega)} + C \norm{u}^2_{L^2(\Omega)}.
	\end{equation}
	The constant $C=C(\Omega,\delta)$ is independent of $u$. 
	\end{lemma}

	\begin{proof} {\bf Proof of \eqref{eq-norm-equivalence}.} Write $s=m+ \sigma$ with $m \in \mathbb{N}$ and $0\leq \sigma <1$.
		We only prove the case $\sigma>0$, namely $s>m$ (the case $s=m$ is easier). Note that by H\"older's inequality in Fourier space we have for all $f\in H^s(\R^d)$: 
		\begin{align*}
		\norm{f}^2_{H^m(\R^d)} &\le C \int_{\R^ d} (1+ |p|^2)^m |\widehat f (p)|^2 d p \\
		&\le C\Big( \int_{\R^ d} (1+ |p|^2)^{s} |\widehat f (p)|^2 d p  \Big)^{\frac{m}{s}}  \Big( \int_{\R^ d} |\widehat f (p)|^2 d p  \Big)^{1-{\frac{m}{s}}}\\
		&\le C \norm{f}_{H^s(\R^d)}^{2\frac{m}{s}}\norm{f}_{L^2(\R^d)}^{2(1-{\frac{m}{s}})}.
		\end{align*}
		Let $T$ be an extension operator, namely $T:H^t(\Omega)\to H^t(\R^d)$ is a bounded linear operator for all $0\le t\le s$ and $Tu_{|\Omega}=u$. For any $u\in H^s(\Omega)$, using the above estimate with $f=Tu$ and Young's Inequality we obtain, for every $\eps>0$, 
		\begin{align*}
		\norm{u}^2_{H^m(\Omega)} & \leq \norm{Tu}^2_{H^m(\R^d)} \\
		&\leq  C\norm{Tu}_{H^s(\R^d)}^{2\frac{m}{s}}\norm{Tu}_{L^2(\R^d)}^{2(1-{\frac{m}{s}})} \\
		&= C(\epsilon\norm{Tu}^2_{H^s(\R^d)})^{\frac{m}{s}} (\epsilon^{-\frac{m}{s-m}} \norm{Tu}^2_{L^2(\R^d)})^{1-{\frac{m}{s}}} \\
		&\leq C( \epsilon \norm{Tu}^2_{H^s(\R^d)} + C_\epsilon \norm{Tu}^2_{L^2(\R^d)}) \\
		& \leq C(\epsilon\norm{u}^2_{H^s(\Omega)} + C_\epsilon \norm{u}^2_{L^2(\Omega)}).
		\end{align*}
		Rearranging the terms we find that 
		$$(1-C\epsilon) \norm{u}_{H^m(\Omega)}^2 \leq C(\epsilon\norm{u}_{\dot{H}^s(\Omega)}^2 + C_\epsilon \norm{u}_{L^2(\Omega)}^2).$$
		By choosing $\epsilon>0$ small enough we arrive at \eqref{eq-norm-equivalence}.

\bigskip
\noindent{\bf Proof of \eqref{eq-norm-equivalence-II}.}  
	Let $T$ be an extension operator. Then for any $\epsilon >0$
	\begin{align*}
	\norm{u}_{\dot{H}^t(\Omega)}^2 & \leq  \norm{Tu}_{\dot{H}^t(\R ^d)}^2 \\
	& \leq C\norm{Tu}_{\dot{H}^s(\R ^d)}^{2\frac{t}{s}}   \norm{Tu}_2^{2(1-\frac{t}{s})} \\
	&= C(\epsilon\norm{Tu}_{\dot{H}^s(\R ^d)}^2)^{\frac{t}{s}}  
	(\epsilon ^{-\frac{t}{s-t}}\norm{Tu}_2^2)^{1-\frac{t}{s}} \\
	& \leq C( \epsilon\norm{Tu}_{\dot{H}^s(\R ^d)}^2 + \epsilon ^{-\frac{t}{s-t}}\norm{Tu}_2^2) \\
	& \leq C (\epsilon\norm{u}_{H^s(\Omega)}^2 + \epsilon ^{-\frac{t}{s-t}}\norm{u}_2^2)
	\end{align*}
	where the second step follows from H\"older's inequality in Fourier space and the third step follows from Young's Inequality.
	Now \eqref{eq-norm-equivalence-II} follows from \eqref{eq-norm-equivalence}.
\end{proof}

Now we provide 

\begin{proof}[Proof of \eqref{eq-local-uncertainty-II-one-body}] {\bf Step 1.} Let $\chi, \eta: \R^d \to [0,1]$ be two smooth functions such that 
$$\chi^2+\eta^2=1, \quad \chi(x)=1 \text{ if } x\in \Omega, \quad {\rm supp}\chi \subset\wO.$$
By the definition 
\begin{align*}
	C_{\rm GN} := \inf_{\substack{u \in H^s(\R^d)\\ 
			\norm{u}_{L^2} = 1}} \frac{\langle u, (-\Delta)^s u \rangle}{\int_{\R^d} |u|^{2 (1+\frac{2s}{d})}}.
  \end{align*}	
   we have
\begin{align} \label{eq:ABC-1}
\|\chi u\|^2_{\dot H^s(\R^d)}  \ge  C_{\rm GN} \frac{\int_{\R^d} |\chi u|^{2(1+2s/d)} }{\Big( \int_{\R^d} |\chi u|^2 \Big)^{2s/d}} \ge C_{\rm GN} \frac{\int_{\Omega} |u|^{2(1+2s/d)} }{\Big( \int_{\wO} |u|^2 \Big)^{2s/d}}.
\end{align}
It remains to compare $\|\chi u\|^2_{\dot H^s(\R^d)}$ with $\| u\|^2_{\dot H^s(\wO)}$. 

\bigskip

\noindent{\bf Step 2.}  Now we prove that for any $\delta\in (0,1)$
\begin{align} \label{eq:ABC-2}
\|\chi u\|^2_{\dot H^s(\R^d)}  \le  (1+\delta) \|\chi u\|^2_{\dot H^s(\wO)} + C_{\delta} \|u\|_{L^2(\wO)}^2. 
\end{align}

If $s \in \mathbb{N}$, then $\|\chi u\|^2_{\dot H^s(\R^d)}= \|\chi u\|^2_{\dot H^s(\wO)}$ since $\supp (\chi u)\subset \wO$ (thus \eqref{eq:ABC-2} is trivial). Consider the case for $s=m + \sigma$ with $m \in \mathbb{N}$ and $0<\sigma<1$. Since ${\rm supp}\chi \subset \subset\wO$ we have 
$$\ell:={\rm dist}({\rm supp} \chi, \R^d\backslash \wO)>0.$$
Therefore, 
	\begin{align*}
	\norm{\chi u}_{\dot{H}^s(\R^d)}^2 &= \norm{\chi u}_{\dot{H}^s(\wO)}^2 + 
	 \sum_{\abs{\alpha}=m}	2 c_{d, \sigma} \frac{m!}{\alpha!}  \int_{\wO} \int_{\R^d\backslash \wO} \frac{\abs{D^{\alpha}(\chi u)(x)-D^{\alpha}(\chi u)(y)}^2}{\abs{x-y}^{d + 2 \sigma}} \,dx dy \\
	&= \norm{\chi u}_{\dot{H}^s(\wO)}^2 + 
	\sum_{\abs{\alpha}=m}	2 c_{d, \sigma} \frac{m!}{\alpha!} \int_{\wO} \int_{\R^d\backslash \wO} \frac{\abs{D^{\alpha}(\chi u)(y)}^2}{\abs{x-y}^{d + 2 \sigma}} \,dx dy \\
	& \leq \norm{\chi u}_{\dot{H}^s(\wO)}^2 + 
	\sum_{\abs{\alpha}=m}	2 c_{d, \sigma} \frac{m!}{\alpha!}  \int_{\wO}\abs{D^{\alpha}(\chi u)(y)}^2 dy \int_{\R^d \backslash B_\ell(0)} \frac{1}{\abs{x}^{d+2\sigma}} \,dx \\
	&= \norm{\chi u}_{\dot{H}^s(\wO)}^2 + C_{d,\sigma,m,\ell}  \norm{\chi u}_{\dot{H}^m(\wO)}^2.
	\end{align*}
	In the second equality we used $D^\alpha(\chi u)(x)=0$ when $x\notin \wO$. 	By \eqref{eq-norm-equivalence-II} we have
	$$
	\norm{\chi u}_{\dot{H}^m(\wO)}^2 \le \delta \norm{\chi u}_{\dot{H}^s(\wO)}^2 + C_\delta \|\chi u\|_{L^2(\wO)}^2.
	$$
	Thus  \eqref{eq:ABC-2} holds true.  	It remains to bound $\|\chi u\|^2_{\dot H^s(\wO)}$ from above. 
\bigskip

\noindent{\bf Step 3.}  Now we prove that for any $\delta\in (0,1)$
\begin{align} \label{eq:ABC-3}
	\norm{\chi u}_{\dot{H}^s(\wO)}^2 \le (1+\delta) \norm{u}_{\dot{H}^s(\wO)}^2+ C_\delta \| u\|_{L^2(\wO)}^2.  
\end{align}

We use the fractional IMS localization formula from \cite[Lemma 14]{LuNaPo-16}:  
$$
	\abs{\norm{u}_{\dot{H}^s(\wO)}^2 - \norm{\chi u}_{\dot{H}^s(\wO)}^2-\norm{\eta u}_{\dot{H}^s(\wO)}^2}
	\leq C \left(\norm{\chi u}_{H^t(\wO)}^2+\norm{\eta u}_{H^t(\wO)}^2 \right)
$$
	for some $t\in (0,s)$. By \eqref{eq-norm-equivalence-II} again we have, for any $\delta\in (0,1)$,
	$$
	\norm{\chi u}_{{H}^t(\wO)}^2 \le \delta \norm{\chi u}_{\dot{H}^s(\wO)}^2 + C_\delta \|\chi u\|_{L^2(\wO)}^2, \quad \norm{\eta u}_{{H}^t(\wO)}^2 \le \delta \norm{\eta u}_{\dot{H}^s(\wO)}^2 + C_\delta \|\eta u\|_{L^2(\wO)}^2.
	$$
	Therefore, for any $\delta\in (0,1)$,
	$$
	\abs{\norm{u}_{\dot{H}^s(\wO)}^2 - \norm{\chi u}_{\dot{H}^s(\wO)}^2-\norm{\eta u}_{\dot{H}^s(\wO)}^2}
	\leq \delta \left(\norm{\chi u}_{H^s(\wO)}^2+\norm{\eta u}_{H^s(\wO)}^2 \right) + C_\delta \|u\|_{L^2(\wO)}^2
$$
and hence
\begin{align*}
	\norm{u}_{\dot{H}^s(\wO)}^2 & \ge (1-\delta) \Big( \norm{\chi u}_{\dot{H}^s(\wO)}^2+\norm{\eta u}_{\dot{H}^s(\wO)}^2 \Big) - C_\delta \| u\|_{L^2(\wO)}^2.
\end{align*}
The latter bound implies \eqref{eq:ABC-3}. The inequality \eqref{eq-local-uncertainty-II-one-body}  follows from \eqref{eq:ABC-1}, \eqref{eq:ABC-2} and \eqref{eq:ABC-3}. 
\end{proof}

\section{Local exclusion} \label{sec5}

We now prove our local exclusion bound which will allow us to control the error terms from the local uncertainty bounds. Here we have to refine 
the local exclusion in \cite[Lemma 6]{LuNaPo-16} (see also \cite[Theorem 2]{LuPoSo-15}) as the existing bound is only good for a cube, and it becomes very weak for   a set with small volume to length ratio, for example a long chain of cubes. In the following, we will deal with general sets which can be decomposed into several pieces with small diameters. Also, we allow that the sets in different length scales may be overlapping.

\begin{lemma}[Local exclusion principle] \label{lem-interaction-K2} Let $d\ge 1$  and $s>0$. Let $\{R_n\}_{n\ge 1}$ be a decreasing sequence of positive numbers. Let $\{\Omega_{n,m}\}_{n,m\ge 1}$ be a collection of subsets of $\R^d$ such that 
$$
{\rm diam}(\Omega_{n,m}) \le R_n, \quad \sum_{m} \1_{\Omega_{n,m}} \le C_n. 
$$
Then for any normalized wave function $\Psi\in L^2(\R^{dN})$ we have
\begin{align}
\Big\langle \Psi, \sum_{1\le i<j\le N} \frac{1}{|x_i-x_j|^{2s} }  \Psi\Big\rangle \ge \sum_{n\ge 1}  \frac{1}{2 C_n} \Big( \frac{1}{R_n^{2s}} - \frac{1}{R_{n-1}^{2s}}\Big)  \sum_{m\ge 1} \Big(\int_{\Omega_{n,m}} \rho_\Psi \Big) \Big( \int_{\Omega_{n,m}} \rho_\Psi -1 \Big). 
\end{align}
Here we use the convention $R_0=+\infty$. 
\end{lemma}

	\begin{proof} We start from an elementary but very useful formula  
	\begin{align*}
	\frac{1}{|x|^{2s}} &\ge \sum_{n\ge 1} 	\frac{1}{|x|^{2s}} \1(R_{n+1}< |x| \le R_n) \\
	&\ge \sum_{n\ge 1} 	\frac{1}{R_n^{2s}} \Big( \1(|x|\le R_n) - \1(|x|\le R_{n+1}) \Big)\\
	&= \sum_{n\ge 1} 	\Big( \frac{1}{R_n^{2s}} - \frac{1}{R_{n-1}^{2s}}\Big)   \1(|x|\le R_n) 
	\end{align*}
	with the convention $R_0=+\infty$. Moreover, from the assumption
	$$
{\rm diam}(\Omega_{n,m}) \le R_n, \quad \sum_{m} \1_{\Omega_{n,m}} \le C_n. 
$$
we can bound 
$$
 \1(|x_i-x_j |\le R_n) \ge \frac{1}{C_n}\sum_{m} \1_{\Omega_{n,m}} (x_i) \1_{\Omega_{n,m}} (x_j).
 $$
Consequently, we have the pointwise estimate
	\begin{align*}
	\frac{1}{|x_i-x_j|^{2s}} &\ge \sum_{n\ge 1} 	\Big( \frac{1}{R_n^{2s}} - \frac{1}{R_{n-1}^{2s}}\Big)   \1(|x_i-x_j |\le R_n) \\
	&\ge \sum_{n\ge 1} \sum_{m \ge 1} \frac{1}{C_n} \Big( \frac{1}{R_n^{2s}} - \frac{1}{R_{n-1}^{2s}}\Big)  \1_{\Omega_{n,m}} (x_i) \1_{\Omega_{n,m}} (x_j). 
	\end{align*}
	Next, similarly to \cite[Lemma 6]{LuNaPo-16},  by the Cauchy--Schwarz inequality we have
	 	\begin{align*}
		\Big\langle \Psi, \sum_{1\le i<j\le N} \1_{\Omega_{n,m}} (x_i) \1_{\Omega_{n,m}} (x_j) \Psi\Big\rangle &= \frac{1}{2} \Big\langle \Psi,  \Big( \sum_{i=1}^N \1_{\Omega_{n,m}} (x_i) \Big)^2  \Psi\Big\rangle  -  \frac{1}{2} \Big\langle \Psi,\sum_{i=1}^N  \1_{\Omega_{n,m}} (x_i) \Psi\Big\rangle \\
		&\ge \frac{1}{2} \Big| \Big\langle \Psi, \sum_{i=1}^N  \1_{\Omega_{n,m}} (x_i)  \Psi \Big\rangle \Big|^2  -\frac{1}{2}\Big\langle \Psi, \sum_{i=1}^N  \1_{\Omega_{n,m}} (x_i)  \Psi \Big\rangle  \\
		& = \frac{1}{2} \Big( \int_{\Omega_{n,m}} \rho_\Psi \Big)^2 - \frac{1}{2} \int_{\Omega_{n,m}} \rho_\Psi .
				\end{align*}
This ends the proof of Lemma \ref{lem-interaction-K2}. 
		\end{proof}

\section{Proof of Theorem \ref{thm-main1}} \label{sec6}
 
 \begin{proof} By a standard density argument, we can assume that the normalized wave function $\Psi\in L^2(\R^{dN})$ is smooth with compact support. By scaling, we can assume that $\Psi$ is supported in $[-1/2,1/2]^{dN}$. Consequently, the one-body density
 $\rho_\Psi $
 is supported in $[-1/2,1/2]^d$. 

  \bigskip
  \noindent 
 {\bf Step 1: A decomposition of covering sub-cubes.}  We fix constants $\delta\in (0,1)$ and  $\eps=n_0^{-1}$ with an integer number $n_0\ge 2$ (we can choose $n_0=2$).  We divide $[-1/2,1/2]^d$ into disjoint sub-cubes by induction:  in the $n$-th step we obtain a collection $G^{n}$ of sub-cubes of side length $\eps^{n}$, which can be decomposed further into three disjoint sub-collections 
  $$G^{n}=G^{n,0}\bigcup G^{n,1}\bigcup G^{n,2}.$$
  Heuristically, $G^{n,0} \bigcup G^{n,1}$ contains ``good sets" concerning the uncertainty principle while $G^{n,2}$ contains ``good sets" concerning the exclusion principle.  The precise construction is as follows. 
  

  \bigskip
  \noindent 
{\bf Initial step.} When $n=0$, we simply take 
$$G^0=G^{0,2}=\{[-1/2,1/2]^d\}, \quad G^{0,0}=G^{0,1}=\emptyset.$$ 

 \bigskip
  \noindent 
{\bf Induction step.} Let 
$$G^{n-1}=G^{n-1,0}\bigcup G^{n-1,1}\bigcup G^{n-1,2}$$
 be the collection of sub-cubes of side length $\eps^{n-1}$ obtained from the $(n-1)$-th step. In the $n$-th step, we divide each sub-cube in $G^{n-1,2}$ into $\eps^{-d}$ sub-cubes of the same size. Thus each new sub-cube has the side length $\eps^{n}$. Let $G^{n}$ be the collection of all these new sub-cubes. We decompose 
 $$G^{n}=G^{n,0}\bigcup G^{n,1}\bigcup G^{n,2}$$
 as follows. 
 \begin{itemize}
 \item  We denote by $G^{n,0}$ the collection of all sub-cubes $Q$ in $G^{n}$ such that 
 $$\int_Q \rho_\Psi\le \delta.$$
 
 \item We can think of the sub-cubes in $G^n\backslash G^{n,0}$ as a graph where we put edges between neighboring sub-cubes (a cube $Q_1$ neighbours to  a cube $Q_2$ if ${\rm dist}(Q_1,Q_2)=0$). Thus the sub-cubes in $G^n\backslash G^{n,0}$ can be decomposed into disjoint connected components that we call {\bf clusters} (here the connectivity is considered in the graphical sense, which is different from the topological sense). For any cluster $K \subset G^n\backslash G^{n,0}$, we define 
\begin{align} \label{eq:closure}
 \Omega_K :=   \bigcup_{Q\in K} Q, \quad \widetilde \Omega_K :=  \Big\{x\in \R^d,  {\rm dist} (x, \Omega_K) <\frac{\eps^n}{4} \Big\}.
 \end{align}
Note that each  {\bf closure} $\widetilde \Omega_K$ is topologically connected and the closures $\{\widetilde \Omega_K\}$ of different clusters $K$ are disjoint (the sub-cubes in $G^{n,0}$ serve to separate these components).

 \item We denote by $G^{n,1}$ the union of all clusters $K$ such that 
 $$
 \int_{  \widetilde \Omega_K} \rho_\Psi <1+\delta,
 $$
 and $G^{n,2}$ the union of all clusters $K$ such that 
   $$
 \int_{  \widetilde \Omega_K} \rho_\Psi \ge 1+\delta.
 $$
 Only the sub-cubes in $G^{n,2}$ will be divided further in the $(n+1)$-step. 
 \end{itemize}

 Since $\rho_{\Psi} \in L^1([0,1]^d)$ the construction terminates after finitely many steps.
	We now have a division of $[0,1]^d$ as the disjoint union of sub-cubes
	\begin{align} \label{eq:inclusion}
	\supp \rho_\Psi \subset [0,1]^d = \bigcup_{n\ge 1} \Big(  \bigcup_{Q\in G^{n,0} \cup G^{n,1}} Q \Big) .
		\end{align}

	\bigskip
  \noindent 
 {\bf Step 2: Uncertainty principle for $G^{n,0}$.}  For any sub-cube $Q\in G^{n,0}$, we have $|Q|=\eps^{nd}$ and
 $$
 \int_Q \rho_\Psi \le \delta. 
 $$
 Therefore, the local uncertainty principle in Lemma \ref{lem:LUP} implies that 
\begin{align} \label{eq:final-1a}
	\Big\langle \Psi, \sum_{i=1}^N (-\Delta_{x_i})^{s}_{|Q} \Psi \Big\rangle &\ge \frac{1}{C} \frac{\int_Q \rho_\Psi^{1+2s/d} }{\Big( \int_Q \rho_\Psi \Big)^{2s/d}}  - \frac{C}{|Q|^{2s/d}} \int_Q \rho_\Psi \nn\\
	&\ge \frac{1}{C\delta^{2s/d}} \int_Q \rho_\Psi^{1+2s/d}  - \frac{C}{\eps^{2sn}} \int_Q \rho_\Psi, 
	\quad \forall Q\in  G^{n,0}.
	\end{align}
Since the sub-cubes $\{Q\}_{Q\in G^{n,0}, n\ge 1}$ are disjoint, we have 
\begin{align} \label{eq:final-1}
	\delta^{s/d}\Big\langle \Psi, \sum_{i=1}^N (-\Delta_{x_i})^{s}_{|\R^d} \Psi \Big\rangle &\ge  \sum_{n\ge 1} \sum_{Q\in G^{n,0}} \delta^{s/d}  \Big\langle \Psi, \sum_{i=1}^N (-\Delta_{x_i})^{s}_{|Q} \Psi \Big\rangle \nn\\
	& \ge \sum_{n\ge 1} \sum_{Q\in G^{n,0}} \Big( \frac{1}{C\delta^{s/d}} \int_Q \rho_\Psi^{1+2s/d}  - \frac{C\delta^{s/d}}{\eps^{2sn}} \int_Q \rho_\Psi \Big). 
\end{align}
		
\bigskip
  \noindent 
 {\bf Step 3: Uncertainty principle for $G^{n,1}$.} Let $n\ge 1$. For any cluster $K\subset G^{n,1}$,  
  $$
 \int_Q \rho_\Psi \ge \delta, \quad \forall Q\in K
 $$
 while 
 $$
 \sum_{Q\in K}  \int_Q \rho_\Psi = \int_{ \Omega_K} \rho_\Psi \le  \int_{  \widetilde \Omega_K} \rho_\Psi <1+\delta.
 $$
 Here $\widetilde \Omega_K$ is the closure defined in \eqref{eq:closure}. Thus the cluster $K$ is the union of at most $(\delta^{-1}+1)$ disjoint sub-cubes. Consequently, the rescaled set 
 $$\eps^{-n}\Omega_K:=\bigcup_{Q\in K} (\eps^{-n} Q)$$
 is the union of at most $(\delta^{-1}+1)$ disjoint unit cubes in $\R^d$. Moreover, these sub-cubes are connected in the graphical sense (recall that a cube $Q_1$ neighbours to  a cube $Q_2$ if ${\rm dist}(Q_1,Q_2)=0$). Therefore, up to translation (such that there exists one cube in $K$ centered at $0$), $\eps^{-n}\Omega_K$ belongs to a finite collection of subsets of $\R^d$ and the collection depends only on $d$ and $\delta$ (but independent of $\eps,n$). By the definition of the closure, we have 
$$
\eps^{-n} \widetilde \Omega_K :=  \Big\{x\in \R^d,  {\rm dist} (x, \eps^{-n} \Omega_K) <\frac{1}{4} \Big\}.
$$ 
 This implies that up to translation $\eps^{-n} \widetilde \Omega_K$ also belongs to a finite collection of subsets of $\R^d$ which depends only on $d$ and $\delta$.

 Now we apply the local uncertainty principle in Lemma \ref{lem:LUP-II} with $\Omega_K\subset\subset  \widetilde \Omega_K \subset \R^d$. Recall that up to translation, $\eps^{-n} \Omega_K$ and $\eps^{-n} \widetilde \Omega_K$ belong to a finite collection of subsets of $\R^d$ which depends only on $d,\delta$.  Since the kinetic operator $(-\Delta)^s$ is translation-invariant, we deduce that the constant $C_{\delta, \eps^{-n} \Omega_K, \eps^{-n} \widetilde \Omega_K}$ in Lemma \ref{lem:LUP-II} depends only on $d,s,\delta$. Combining with the bound $\int_{ \widetilde \Omega_K} \rho_\Psi <1+\delta$ we get   
 \begin{align} \label{eq:final-2a}
	\Big\langle \Psi, \sum_{i=1}^N (-\Delta_{x_i})^{s}_{| \widetilde \Omega_K} \Psi \Big\rangle &\ge C_{\rm GN} (1-\delta)\frac{ \int_{\Omega_K} \rho_\Psi^{1+2s/d} }{\Big( \int_{ \widetilde \Omega_K} \rho_\Psi \Big)^{2s/d}}  - C_{\delta,\Omega_K, \widetilde \Omega_K} \int_{ \widetilde \Omega_K}\rho_\Psi \nn\\
	&\ge C_{\rm GN} \frac{ (1-\delta)}{(1+\delta)^{2s/d}}  \int_{\Omega_K} \rho_\Psi^{1+2s/d}  - \frac{C_\delta}{\eps^{2sn}} \int_{ \widetilde \Omega_K}\rho_\Psi. 
	\end{align}  
	Since the sets $\{ \wO_K\}_{n\ge 1, K \subset G^{n,1}}$ are disjoint, we find that 
	 \begin{align} \label{eq:final-2}
	 &(1-\delta^{s/d})\Big\langle \Psi, \sum_{i=1}^N (-\Delta_{x_i})^{s}_{| \R^d} \Psi \Big\rangle\ge \sum_{n\ge 1} \sum_{K \subset G^{n,1}} (1-\delta^{s/d}) \Big\langle \Psi, \sum_{i=1}^N (-\Delta_{x_i})^{s}_{| \widetilde \Omega_K} \Psi \Big\rangle 
	\nn\\
	&\ge \sum_{n\ge 1} \sum_{K \subset G^{n,1}} \Big( C_{\rm GN} \frac{ (1-\delta)(1-\delta^{s/d})}{(1+\delta)^{2s/d}}  \int_{\Omega_K} \rho_\Psi^{1+2s/d}  - \frac{C_\delta}{\eps^{2sn}} \int_{ \widetilde \Omega_K}\rho_\Psi \Big). 
	\end{align}  
	Here the sum is taken over all clusters $K \subset G^{n,1}$. 
 
	\bigskip
	  \noindent 
 {\bf Step 4: Local exclusion principle for $G^{n,2}$.} Let $n\ge 1$. For any sub-cube $Q\in G^{n,2}$, we denote $c_Q$ the center of $Q$. Our key observation is that 
\begin{align} \label{eq:mass-Q-Rn}
 \int_{B(c_Q,R_n/2)} \rho_\Psi \ge 1+\delta, \quad \forall Q \in G^{n,2} 
\end{align}
where
 $$R_n:= 2\sqrt{d}( \delta^{-1}+2)\eps^n.$$
 Indeed, any sub-cube $Q\in G^{n,2}$ must belong to a cluster $K$. The set $\Omega_K$ is a connected union of sub-cubes of diameter $\sqrt{d} \eps^n$. Therefore, the ball $B(c_Q,R_n/2)$   contains either the whole set $\wO_K$, or at least $(\delta^{-1}+1)$ disjoint sub-cubes. Thus \eqref{eq:mass-Q-Rn} follows from the facts that 
   $$
 \int_{  \widetilde \Omega_K} \rho_\Psi \ge 1+\delta. 
 $$
 and
   $$
 \int_{Q'} \rho_\Psi \ge \delta, \quad \forall Q'\in K. 
 $$
 
 On the other hand, since the sub-cubes in $G^{n,2}$ are disjoint, the distances of the centers of the sub-cubes are at least $\sqrt{d}\eps^n$. Therefore, 
\begin{align} \label{eq:mass-Q-Rn-b}
\sum_{Q\in K} \1_{B(c_Q,R_n/2)} \le C_\delta
\end{align}
for a constant $C_\delta>0$ depending only on $d,\delta$. 

Now we apply the local exclusion principle in Lemma \ref{lem-interaction-K2} for the balls $\{B(c_Q, R_n/2)\}_{Q\in G^{n,2}}$ with $n=0,1,2,...$ Using 
$$
{\rm diam} B(c_Q, R_n/2)= R_n= 2\sqrt{d}( \delta^{-1}+2)\eps^n
$$
together with \eqref{eq:mass-Q-Rn} and \eqref{eq:mass-Q-Rn-b} we obtain
\begin{align} \label{eq:final-3aa}
&\Big\langle \Psi, \sum_{1\le i<j\le N} \frac{1}{|x_i-x_j|^{2s} }  \Psi\Big\rangle \nn\\
&\ge  \sum_{n\ge 1} \frac{1}{2 C_\delta} \Big( \frac{1}{R_n^{2s}} - \frac{1}{R_{n-1}^{2s}}\Big) \sum_{Q\in G^{n,2}}  \Big(\int_{B(c_Q, R_n/2)} \rho_\Psi \Big) \Big( \int_{B(c_Q, R_n/2)} \rho_\Psi -1 \Big) \nn\\
&\ge \sum_{n\ge 1}  \frac{1}{C_\delta \eps^{2sn}} \sum_{Q\in G^{n,2}} \int_{B(c_Q, R_n/2)} \rho_\Psi 
\end{align}
where the constant $C_\delta$ depends only on $d,s,\delta$. 

Next, recall that for any $n\ge 1$ the sub-cubes in $G^{n,2}$ will be divided further to get the smaller sub-cubes in $G^{n+1,0}$, $G^{n+1,1}$, $G^{n+1,2}$. Consequently, 
$$
\bigcup_{Q\in G^{n,2}} B(c_Q, R_n/2) \supset \Big( \bigcup_{Q\in G^{n+1,0}} Q\Big)   \bigcup \Big( \bigcup_{K\subset G^{n+1,1}} \wO_K \Big) .$$
Here the last union is taken over all clusters $K\subset G^{n+1,1}$. Recall that all sub-cubes in  $G^{n+1,0}$ are disjoint, and all the closures $\wO_K$ of the clusters $K\subset G^{n+1,1}$ are disjoint. Therefore,
 $$
 \sum_{Q\in G^{n,2}} \int_{B(c_Q, R_n/2)} \rho_\Psi \ge \max\Big\{ \sum_{Q\in G^{n+1,0}} \int_Q \rho_\Psi ,  \sum_{K \subset G^{n+1,1}} \int_{\wO_K} \rho_\Psi \Big\}. 
 $$ 
Hence, we deduce from \eqref{eq:final-3aa} that 
$$
\Big\langle \Psi, \sum_{1\le i<j\le N} \frac{1}{|x_i-x_j|^{2s} }  \Psi\Big\rangle \ge \sum_{n\ge 1} \frac{1}{C_\delta \eps^{2sn}} \max\Big\{ \sum_{Q\in G^{n+1,0}} \int_Q \rho_\Psi ,  \sum_{K \subset G^{n+1,1}} \int_{\wO_K} \rho_\Psi \Big\}$$
for a constant $C_\delta>0$ depending only on $d,s,\delta$ but independent of $n$. By shifting $n\mapsto n-1$ and redefine $C_\delta$, we obtain 
$$
\Big\langle \Psi, \sum_{1\le i<j\le N} \frac{1}{|x_i-x_j|^{2s} }  \Psi\Big\rangle \ge \sum_{n\ge 2} \frac{1}{C_\delta \eps^{2sn}} \max\Big\{ \sum_{Q\in G^{n,0}} \int_Q \rho_\Psi ,  \sum_{K \subset G^{n,1}} \int_{\wO_K} \rho_\Psi \Big\}$$ 
for a constant $C_\delta>0$ depending only on $d,s,\eps,\delta$ but independent of $n$. Moreover, since $\supp \Psi_N\subset [0,1]^{dN}$ and $N\ge 2$, we  have the obvious bound 
$$
\Big\langle \Psi, \sum_{1\le i<j\le N} \frac{1}{|x_i-x_j|^{2s} }  \Psi\Big\rangle \ge \frac{N(N-1)}{2} \ge \frac 1 2 \max\Big\{ \sum_{Q\in G^{1,0}} \int_Q \rho_\Psi ,  \sum_{K \subset G^{1,1}} \int_{\wO_K} \rho_\Psi \Big\}.
$$
Thus in summary, we have the local exclusion bound
\begin{align} \label{eq:final-3}
\Big\langle \Psi, \sum_{1\le i<j\le N} \frac{1}{|x_i-x_j|^{2s} }  \Psi\Big\rangle \ge \sum_{n\ge 1} \frac{1}{C_\delta \eps^{2sn}} \Big( \sum_{Q\in G^{n,0}} \int_Q \rho_\Psi +  \sum_{K \subset G^{n,1}} \int_{\wO_K} \rho_\Psi \Big)
\end{align}
for a constant $C_\delta>0$ depending only on $d,s,\eps,\delta$ but independent of $n$ (the value of $C_\delta$ has been changed from line to line). Here the last sum is taken over all clusters $K\subset G^{n,1}$.

 \bigskip
  \noindent 
 {\bf Step 5: Conclusion.} By summing  \eqref{eq:final-1}, \eqref{eq:final-2} and \eqref{eq:final-3} we have
\begin{align*}
  &\Big\langle \Psi, \sum_{i=1}^N (-\Delta_{x_i})^s_{\R^d}  \Psi\Big\rangle +  \lambda \sum_{i=1}^N \Big\langle \Psi, \sum_{1\le i<j\le N} \frac{1}{|x_i-x_j|^{2s} }  \Psi\Big\rangle\\
  & \ge \sum_{n\ge 1} \sum_{Q\in G^{n,0}} \Big( \frac{1}{C\delta^{s/d}} \int_Q \rho_\Psi^{1+2s/d}  - \frac{C\delta^{s/d}}{\eps^{2sn}} \int_Q \rho_\Psi \Big)\\
  & \quad + \sum_{n\ge 1} \sum_{K \subset G^{n,1}} \Big( C_{\rm GN} \frac{ (1-\delta)(1-\delta^{s/d})}{(1+\delta)^{2s/d}}  \int_{\Omega_K} \rho_\Psi^{1+2s/d}  - \frac{C_\delta}{\eps^{2sn}} \int_{ \widetilde \Omega_K}\rho_\Psi \Big). \\
  & \quad + \sum_{n\ge 1} \frac{\lambda}{C_\delta \eps^{2sn}} \Big( \sum_{Q\in G^{n,0}} \int_Q \rho_\Psi +  \sum_{K \subset G^{n,1}} \int_{\wO_K} \rho_\Psi \Big)
\\
&= \sum_{n\ge 1} \Big( \sum_{Q\in G^{n,0}} \frac{1}{C\delta^{s/d}} \int_Q \rho_\Psi^{1+2s/d} + C_{\rm GN} \frac{ (1-\delta)(1-\delta^{s/d})}{(1+\delta)^{2s/d}} \sum_{K\subset G^{n,1}}\int_{\Omega_K} \rho_\Psi^{1+2s/d}   \Big) \\
&\quad +  \sum_{n\ge 1}  \Big(  \frac{\lambda}{C_\delta } - C\delta^{s/d} \Big) \frac{1}{\eps^{2sn} } \sum_{Q\in G^{n,0}}\int_Q \rho_\Psi + \sum_{n\ge 1}  \Big(  \frac{\lambda}{C_\delta} - C_\delta \Big) \frac{1}{\eps^{2sn} } \sum_{K \subset  G^{n,1}}\int_{\wO_K} \rho_\Psi .
 \end{align*}
 For any given $\delta\in (0,1)$, we can choose $\lambda>0$ sufficiently large such that
 $$
 \frac{\lambda}{C_\delta } - C\delta^{s/d}  \ge 0, \quad  \frac{\lambda}{C_\delta } - C_\delta  \ge 0.
 $$
 Then the above estimate reduces to 
 \begin{align*}
  &\Big\langle \Psi, \sum_{i=1}^N (-\Delta_{x_i})^s_{\R^d}  \Psi\Big\rangle +  \lambda \sum_{i=1}^N \Big\langle \Psi, \sum_{1\le i<j\le N} \frac{1}{|x_i-x_j|^{2s} }  \Psi\Big\rangle \\
  & \ge  \sum_{n\ge 1} \Big( \sum_{Q\in G^{n,0}} \frac{1}{C\delta^{s/d}} \int_Q \rho_\Psi^{1+2s/d} + C_{\rm GN} \frac{ (1-\delta)(1-\delta^{s/d})}{(1+\delta)^{2s/d}} \sum_{K\subset G^{n,1}}\int_{\Omega_K} \rho_\Psi^{1+2s/d}   \Big)\\
  &\ge \min\Big\{\frac{1}{C\delta^{s/d}}, C_{\rm GN} \frac{ (1-\delta)(1-\delta^{s/d})}{(1+\delta)^{2s/d}} \Big\}  \sum_{n\ge 1} \Big( \sum_{Q\in G^{n,0}} \int_Q \rho_\Psi^{1+2s/d} + \sum_{K\subset G^{n,1}}\int_{\Omega_K} \rho_\Psi^{1+2s/d}   \Big)\\
  &= \min\Big\{\frac{1}{C\delta^{s/d}}, C_{\rm GN} \frac{ (1-\delta)(1-\delta^{s/d})}{(1+\delta)^{2s/d}} \Big\}  \int_{\R^d} \rho_\Psi^{1+2s/d} .
 \end{align*}
 In the last equality we have used the covering property \eqref{eq:inclusion}.

Thus for every $\delta\in (0,1)$, with $\lambda>0$ sufficiently large we have 
$$
C_{\rm LT}(\lambda) \ge \min\Big\{\frac{1}{C\delta^{s/d}}, C_{\rm GN} \frac{ (1-\delta)(1-\delta^{s/d})}{(1+\delta)^{2s/d}} \Big\}. 
$$
This implies that
$$
\liminf_{\lambda\to \infty} C_{\rm LT}(\lambda)  \ge C_{\rm GN}. 
$$
Together with the known upper bound $C_{\rm LT}(\lambda)\le C_{\rm GN}$ (see \cite[Proposition 10]{LuNaPo-16}), we conclude that $C_{\rm LT}(\lambda) \to  C_{\rm GN}$ as $\lambda\to \infty$. This completes the proof of Theorem \ref{thm-main1}. 
\end{proof}

\section{Proof of Theorem \ref{thm-main2}} \label{sec7}

First, we adapt the local uncertainty principle for Hardy operator $(-\Delta_{x_i})^s - \mathcal{C}_{s,d} |x|^{-2s}$.

\begin{lemma}[Local uncertainty principle for Hardy operator] \label{lem:LUP-H} Let $d \geq 1$, $s>0$. Let $\Psi$ be a normalized wave function in $L^2(\R^{dN})$. Then for any cube $Q \subset \R^d$ centered at $0$ we have
\begin{align} \label{eq:uncertainty-H-1}
	\Big\langle \Psi, \sum_{i=1}^N ((-\Delta_{x_i})^s_{|Q}  - \mathcal{C}_{s,d} |x_i|^{-2s}\1_{Q}(x_i))\Psi \Big\rangle \ge \frac{1}{C} \frac{\int_Q \rho_\Psi^{1+2s/d} }{\Big( \int_Q \rho_\Psi \Big)^{2s/d}}  - \frac{C}{|Q|^{2s/d}} \int_Q \rho_\Psi. 
\end{align}
The constant $C$ is indepdent of $Q,\Psi,N$. Moreover, for two domains $\Omega \subset\subset \widetilde \Omega \subset \R^d$ where $\wO$ is a $s$-extension domain, we have
\begin{align} \label{eq:uncertainty-H-2}
	\Big\langle \Psi, \sum_{i=1}^N ((-\Delta_{x_i})^s_{|\wO} - \mathcal{C}_{s,d} |x_i|^{-2s}\1_{\Omega}(x_i) ) \Psi \Big\rangle \ge C_{\rm HGN} (1-\delta)\frac{\int_\Omega \rho_\Psi^{1+2s/d} }{\Big( \int_{\wO} \rho_\Psi \Big)^{2s/d}}  - C_{\delta,\Omega,\wO} \int_{\wO}\rho_\Psi
\end{align}
for any $\delta\in (0,1)$. The constant $C_{\delta,\Omega,\wO}>0$ is independent of $\Psi,N$ and it scales as 
	$$
	C_{\delta,\Omega,\wO}= C_{\delta,L\Omega,L\wO} L^{2s}, \quad \forall L>0. 
	$$
\end{lemma}

\begin{proof} The first bound \eqref{eq:uncertainty-H-1} is taken from \cite{LuNaPo-16} (\cite[Lemma 13]{LuNaPo-16} contains the one-body version and the $N$-body version follows from a general argument explained in the proof of Lemma \ref{lem:LUP-II}). 

For the second bound \eqref{eq:uncertainty-H-2}, by following the proof of Lemma \ref{lem:LUP-II}, we only need to prove the following one-body counterpart: for any $u\in H^s(\wO)$  and $\delta\in (0,1)$,
	\begin{equation} \label{eq-local-uncertainty-II-one-body-H}
	\|u\|_{\dot{H}^s(\wO)}^2 - \mathcal{C}_{s,d} \int_{\Omega} \frac{|u(x)|^2}{|x|^{2s}} dx  \ge C_{\rm HGN} (1-\delta)\frac{\int_\Omega |u|^{2(1+2s/d)} }{\Big( \int_{\wO} |u|^2 \Big)^{2s/d}}  - C_{\delta,\Omega,\wO} \int_{\wO} |u|^2. 
	\end{equation}
Note that the scaling property $C_{\delta,\Omega,\wO}=C_{\delta,L\Omega,L\wO} L^{2s}$ follows from a simple change of variables. In the following, we will prove \eqref{eq-local-uncertainty-II-one-body-H} for fixed $(\Omega,\wO)$, and hence we will write $C_\delta$ instead of $C_{\delta,\Omega,\wO}$ for simplicity. 
\bigskip

\noindent
{\bf Step 1.} We start by proceeding as in the proof of   \eqref{eq-local-uncertainty-II-one-body}. Let $\chi, \eta: \R^d \to [0,1]$ be two smooth functions such that 
$$\chi^2+\eta^2=1, \quad \chi(x)=1 \text{ if } x\in \Omega, \quad {\rm supp}\chi \subset\wO.$$
By the definition 
\begin{align*}
	C_{\rm HGN} := \inf_{\substack{u \in H^s(\R^d)\\ 
			\norm{u}_{L^2} = 1}} \frac{\Big\langle u, \Big((-\Delta)^s - \mathcal{C}_{s,d} |x|^{-2s} \Big) u \Big\rangle}{\int_{\R^d} |u|^{2 (1+\frac{2s}{d})}}.
  \end{align*}	  
  we have
\begin{align} \label{eq:ABC-H-1a}
\|\chi u\|^2_{\dot H^s(\R^d)} - \mathcal{C}_{s,d} \int_{\R^d} \frac{|\chi u|^2}{|x|^{2s}} d x \ge  C_{\rm HGN} \frac{\int_{\R^d} |\chi u|^{2(1+2s/d)} }{\Big( \int_{\R^d} |\chi u|^2 \Big)^{2s/d}},\quad \forall u\in H^s(\R^d).
\end{align}
Moreover, we will use the following powerful improvement of Hardy's inequality: for all $s>t>0$ and $\ell>0$,
$$
 (-\Delta)^s - \frac{\mathcal{C}_{s,d}}{|x|^{2s}} \ge  \ell^{s-t}(-\Delta)^t -C_{d,s,t}\ell^{s} \quad \text{on}~ L^2(\R^d).
$$
This bound was first proved for $s=1/2$, $d=3$ 
by Solovej, S\o rensen and Spitzer  \cite[Lemma 11]{SolSorSpi-10} and then 
generalized to the full range $0<s<d/2$ by Frank \cite[Theorem 1.2]{Frank-09}. Consequently,  for any fixed $s>t>0$ and  $\delta \in (0,1)$ we have
\begin{align} \label{eq:ABC-H-1b}
\delta\Big( \|\chi u\|^2_{\dot H^s(\R^d)} - \mathcal{C}_{s,d} \int_{\R^d} \frac{|\chi u|^2}{|x|^{2s}} d x \Big) \ge \delta^{-1}  \|\chi u\|^2_{\dot H^t(\R^d)} - C_{t,\delta} \| \chi u\|_{L^2(\R^d)}^2.
\end{align}
Multiplying \eqref{eq:ABC-H-1a} with $(1-\delta)$ and then summing with \eqref{eq:ABC-H-1b}, we deduce that for any fixed $s>t>0$ and  $\delta \in (0,1)$, 
\begin{align*} 
 \|\chi u\|^2_{\dot H^s(\R^d)} - \mathcal{C}_{s,d} \int_{\R^d} \frac{|\chi u|^2}{|x|^{2s}} d x &\ge (1-\delta) C_{\rm HGN} \frac{\int_{\R^d} |\chi u|^{2(1+2s/d)} }{\Big( \int_{\R^d} |\chi u|^2 \Big)^{2s/d}} + \delta^{-1} \|\chi u\|^2_{\dot H^t(\R^d)} \nn\\
 &\qquad - C_{t,\delta} \| \chi u\|_{L^2(\R^d)}^2, \quad \forall u\in H^s(\R^d).
\end{align*}
Since $\1_\Omega\le \chi \le \1_{\wO}$  the latter estimate reduces to  
\begin{align} \label{eq:ABC-H-1}
 \|\chi u\|^2_{\dot H^s(\R^d)} - \mathcal{C}_{s,d} \int_{\Omega} \frac{|\chi u|^2}{|x|^{2s}} d x &\ge (1-\delta) C_{\rm HGN} \frac{\int_{\Omega} | u|^{2(1+2s/d)} }{\Big( \int_{\wO} |\chi u|^2 \Big)^{2s/d}} + \delta^{-1} \|\chi u\|^2_{\dot H^t(\R^d)} \nn\\
 &\qquad - C_{t,\delta} \| u\|_{L^2(\wO)}^2, \quad \forall u\in H^s(\R^d).
\end{align}

\bigskip

\noindent
{\bf Step 2.} Now let us compare $\|\chi u\|^2_{\dot H^s(\R^d)}$ with $\|u\|^2_{\dot H^s(\wO)}$. 

We can compare $\|\chi u\|^2_{\dot H^s(\R^d)}$ with $\|\chi u\|^2_{\dot H^s(\wO)}$ as in the proof of \eqref{eq:ABC-2}. Recall that if $s\in \mathbb{N}$, then $\|\chi u\|^2_{\dot H^s(\R^d)}=\|\chi u\|^2_{\dot H^s(\wO)}$ since $\supp (\chi u)\in \wO$. If $s=m+\sigma$ with $m\in \mathbb{N}$ and $0<\sigma<1$, then from proof of \eqref{eq:ABC-2}, we find that 
\begin{align*}
	 \norm{\chi u}_{\dot{H}^s(\wO)}^2 \le \norm{\chi u}_{\dot{H}^s(\R^d)}^2 \le \norm{\chi u}_{\dot{H}^s(\wO)}^2 + C \norm{\chi u}_{\dot{H}^m(\wO)}^2 =  \norm{\chi u}_{\dot{H}^s(\wO)}^2 + C \norm{\chi u}_{\dot{H}^m(\R^d)}^2. 
	\end{align*}
	Thus in summary, for any $s>0$ we can find $0\leq t_1<s$ such that
	\begin{align} \label{eq:ABC-H-2a}
	\norm{\chi u}_{\dot{H}^s(\wO)}^2 \le \norm{\chi u}_{\dot{H}^s(\R^d)}^2 \le \norm{\chi u}_{\dot{H}^s(\wO)}^2 + C \norm{\chi u}_{\dot{H}^{t_1}(\R^d)}^2.
	\end{align}

Next, we compare $\|\chi u\|^2_{\dot H^s(\wO)}$ with $\|u\|^2_{\dot H^s(\wO)}$ as in the proof of \eqref{eq:ABC-3}. Recall that by the IMS formula in \cite[Lemma 14]{LuNaPo-16}, we can find $0<t_2<s$ such that  
$$
	\abs{\norm{u}_{\dot{H}^s(\wO)}^2 - \norm{\chi u}_{\dot{H}^s(\wO)}^2-\norm{\eta u}_{\dot{H}^s(\wO)}^2}
	\leq C \left(\norm{\chi u}_{H^{t_2}(\wO)}^2+\norm{\eta u}_{H^{t_2}(\wO)}^2 \right). 
$$
	Moreover, thanks to \eqref{eq-norm-equivalence-II} we can estimate further 
	$$
	C \norm{\eta u}_{H^{t_2}(\wO)}^2 \le \norm{\eta u}_{\dot{H}^s(\wO)}^2 +  C' \norm{\eta u}_{L^2(\wO)}^2
	$$
	and
	$$
	C\norm{\chi u}_{{H}^{t_2}(\wO)}^2 \le C' \norm{\chi u}_{\dot{H}^{t_2}(\wO)}^2 + C'  \norm{\chi u}_{L^2(\wO)}^2 \le C ' \norm{\chi u}_{\dot{H}^{t_2}(\R^d)}^2 + C'  \norm{\chi u}_{L^2(\wO)}^2. 
	$$
Thus
$$
	\abs{\norm{u}_{\dot{H}^s(\wO)}^2 - \norm{\chi u}_{\dot{H}^s(\wO)}^2-\norm{\eta u}_{\dot{H}^s(\wO)}^2}
	\leq  \norm{\eta u}_{\dot{H}^s(\wO)}^2  + C \norm{\chi u}_{\dot{H}^{t_2}(\R^d)}^2 + C  \norm{u}_{L^2(\wO)}^2
$$
for a constant $C$ independent of $u$. By the triangle inequality, we find that 
	\begin{align} \label{eq:ABC-H-2b}
	\norm{\chi u}_{\dot{H}^s(\wO)}^2 \le \norm{u}_{\dot{H}^s(\wO)}^2 + C \norm{\chi u}_{\dot{H}^{t_2}(\R^d)}^2 + C  \norm{u}_{L^2(\wO)}^2. 
	\end{align}
	Combining \eqref{eq:ABC-H-2a} and \eqref{eq:ABC-H-2b} we conclude that
	\begin{align} \label{eq:ABC-H-2}
	\norm{\chi u}_{\dot{H}^s(\R^d)}^2 \le \norm{u}_{\dot{H}^s(\wO)}^2 + C\norm{\chi u}_{\dot{H}^{t_1}(\R^d)}^2  + C\norm{\chi u}_{\dot{H}^{t_2}(\R^d)}^2 + C \| u\|_{L^2(\wO)}^2
	\end{align}
	for some constants $t_1,t_2\in (0,s)$. The constant $C$ is independent of $u$.

\bigskip

\noindent
{\bf Step 3.} Finally, we deduce from \eqref{eq:ABC-H-1} and \eqref{eq:ABC-H-2} that 
\begin{align} \label{eq:ABC-H-3a}
 \|u\|^2_{\dot H^s(\wO)} - \mathcal{C}_{s,d} \int_{\Omega} \frac{|\chi u|^2}{|x|^{2s}} d x &\ge (1-\delta) C_{\rm HGN} \frac{\int_{\Omega} |u|^{2(1+2s/d)} }{\Big( \int_{\wO} |\chi u|^2 \Big)^{2s/d}} + \delta^{-1} \|\chi u\|^2_{\dot H^t(\R^d)}\nn\\
 &\qquad - C\norm{\chi u}_{\dot{H}^{t_1}(\R^d)}^2  - C\norm{\chi u}_{\dot{H}^{t_2}(\R^d)}^2  - C_{t,\delta} \| u\|_{L^2(\wO)}^2. 
\end{align}
This bound holds for all $t\in (0,s)$. Therefore, we can choose $\max\{t_1,t_2\} < t<s$ and use the pointwise estimate 
$$ \delta^{-1} |p|^{2t} -C|p|^{2t_1} - C|p|^{2t_2} \ge - C_\delta, \quad \forall p\in \R^d$$
in the Fourier space to get
$$
\delta^{-1} \|\chi u\|^2_{\dot H^t(\R^d)} - C\norm{\chi u}_{\dot{H}^{t_1}(\R^d)}^2  - C\norm{\chi u}_{\dot{H}^{t_2}(\R^d)}^2 \ge - C_\delta \|\chi u\|_{L^2(\R^d)}^2 \ge - C_\delta \|u\|_{L^2(\wO)}^2. 
$$ 
Thus \eqref{eq:ABC-H-3a} reduces to 
\begin{align} \label{eq:ABC-H-3b}
 \|u\|^2_{\dot H^s(\wO)} - \mathcal{C}_{s,d} \int_{\Omega} \frac{|\chi u|^2}{|x|^{2s}} d x &\ge (1-\delta) C_{\rm HGN} \frac{\int_{\Omega} |u|^{2(1+2s/d)} }{\Big( \int_{\wO} |\chi u|^2 \Big)^{2s/d}} - C_{\delta} \| u\|_{L^2(\wO)}^2
\end{align}
for all $\delta \in (0,1)$ and $u\in H^s(\R^d)$. The constant $C_\delta$ depends on $\delta, \Omega,\wO$, but it is independent of $u$. Thus \eqref{eq-local-uncertainty-II-one-body-H} holds true. This completes the proof of Lemma \ref{lem:LUP-H}. 
\end{proof}

We are ready to provide

\begin{proof}[Proof of Theorem \ref{thm-main2}] Again we can assume that the normalized wave function $\Psi\in L^2(\R^{dN})$ is smooth and supported in $[0,1]^{dN}$. 

  \bigskip
  \noindent 
 {\bf Step 1: Covering sub-cubes.}  We fix constants $\delta\in (0,1)$ and  $\eps=n_0^{-1}$ with an {\em odd} integer number $n_0\ge 3$ (we can choose $n_0=3$).  We construct the collections of sub-cubes 
  $$G^{n}=G^{n,0}\bigcup G^{n,1}\bigcup G^{n,2}$$
exactly as in the proof of Theorem \ref{thm-main1}. Thus as in \eqref{eq:inclusion}, $\supp \rho_\Psi$ is covered by disjoint sub-cubes: 
\begin{align} \label{eq:inclusion-1}
	\supp \rho_\Psi \subset [0,1]^d = \bigcup_{n\ge 1} \Big(  \bigcup_{Q\in G^{n,0} \cup G^{n,1}} Q \Big) .
\end{align}
The choice $\eps=n_0^{-1}$ with $n_0$ odd gives us an additional property: for any sub-cube $Q\in G^{n,0} \cup G^{n,1}$, either  $0$ is the center of $Q$, or 
\begin{align}
{\rm dist}(0,Q) \ge \frac{|Q|^{1/d}}{2}= \frac{\eps^n}{2}. \label{eq:dist-0-Q}
\end{align}

	\bigskip
  \noindent 
 {\bf Step 2: Uncertainty principle I.}  We prove that for any sub-cube $Q\in G^{n,0} \cup G^{n,1}$, 
\begin{align} \label{eq:H-final-1a}
	\Big\langle \Psi, \sum_{i=1}^N \Big( (-\Delta_{x_i})^{s}_{|Q} -\mathcal{C}_{s,d} |x_i|^{-2s} \1_Q(x_i) \Big) \Psi \Big\rangle &\ge \frac{1}{C} \frac{\int_Q \rho_\Psi^{1+2s/d} }{\Big( \int_Q \rho_\Psi \Big)^{2s/d}}  - \frac{C}{|Q|^{2s/d}} \int_Q \rho_\Psi.
	\end{align}
	Indeed,  if $0$ is the center of $Q$, then \eqref{eq:H-final-1a} is exactly the first bound \eqref{eq:uncertainty-H-1} in Lemma \ref{lem:LUP-H}.  Otherwise, if $0\notin Q$, then using \eqref{eq:dist-0-Q} we have
\begin{align*} 
	\Big\langle \Psi, \sum_{i=1}^N \Big( (-\Delta_{x_i})^{s}_{|Q} -\mathcal{C}_{s,d} |x_i|^{-2s} \1_Q(x_i) \Big) \Psi \Big\rangle &=  \Big\langle \Psi, \sum_{i=1}^N (-\Delta_{x_i})^{s}_{|Q}  \Psi \Big\rangle - \mathcal{C}_{s,d}  \int_{Q} \frac{\rho_\Psi(x)}{|x|^{2s}} dx \\
	&\ge \Big\langle \Psi, \sum_{i=1}^N (-\Delta_{x_i})^{s}_{|Q}  \Psi \Big\rangle - \mathcal{C}_{s,d}  \frac{2^{2s}}{ |Q|^{2s/d}} \int_{Q}  \rho_\Psi, 
	\end{align*}
and hence \eqref{eq:H-final-1a} follows from Lemma \ref{lem:LUP}. Thus \eqref{eq:H-final-1a} always holds true. 
	
From \eqref{eq:H-final-1a} and the covering property \eqref{eq:inclusion-1} we find that 
\begin{align*} 
	&\Big\langle \Psi, \sum_{i=1}^N \Big( (-\Delta_{x_i})^{s}_{|\R^d} - \mathcal{C}_{s,d}|x_i|^{-2s} \Big) \Psi \Big\rangle \\
	&\ge  \sum_{n\ge 1} \sum_{Q\in G^{n,0}\cup G^{n,1}}  \Big\langle \Psi, \sum_{i=1}^N \Big( (-\Delta_{x_i})^{s}_{|Q} -  \mathcal{C}_{s,d}|x_i|^{-2s} \1_Q(x_i)\Big) \Psi \Big\rangle \nn\\
	& \ge \sum_{n\ge 1} \sum_{Q\in G^{n,0}\cup G^{n,1}} \Big( \frac{1}{C} \frac{\int_Q \rho_\Psi^{1+2s/d} }{\Big( \int_Q \rho_\Psi \Big)^{2s/d}}  - \frac{C}{|Q|^{2s/d}} \int_Q \rho_\Psi \Big). 
\end{align*}
Recall that  $|Q|=\eps^n$ for all $Q\in G^{n,0}\cup G^{n,1}$. Moreover, if $Q\in G^{n,0}$, then  $\int_Q\rho_\Psi \le \delta$. Therefore, we obtain the lower bound
\begin{align} \label{eq:H-final-1}
	 \Big\langle \Psi, \sum_{i=1}^N \Big( (-\Delta_{x_i})^{s}_{|\R^d} - \mathcal{C}_{s,d}|x_i|^{-2s} \Big) \Psi \Big\rangle &\ge \sum_{n\ge 1} \sum_{Q\in G^{n,0}} \frac{1}{C\delta^{2s/d}}  \int_Q \rho_\Psi^{1+2s/d} \nn\\
	 &\qquad -  \sum_{n\ge 1} \sum_{Q\in G^{n,0}\cup G^{n,1}}  \frac{C}{\eps^{2sn}} \int_Q \rho_\Psi . 
\end{align}

	\bigskip
  \noindent 
 {\bf Step 3: Uncertainty principle II.}  We prove that for any cluster $K \subset G^{n,1}$,  
 \begin{align} \label{eq:H-final-2aa}
	&\Big\langle \Psi, \sum_{i=1}^N \Big( (-\Delta_{x_i})^{s}_{|\wO_K} -\mathcal{C}_{s,d} |x_i|^{-2s} \1_{\Omega_K}(x_i) \Big) \Psi \Big\rangle \nn\\
	&\ge C_{\rm HGN} (1-\delta) \frac{\int_{\Omega_K} \rho_\Psi^{1+2s/d} }{\Big( \int_{\wO_K} \rho_\Psi \Big)^{2s/d}}  - \frac{C_\delta}{\eps^{2sn}} \int_{\wO_K} \rho_\Psi.
	\end{align}
We distinguish two cases. 

{\bf Case 1.} If $0\in \Omega_K$, then as argued  in the proof of Theorem \ref{thm-main1} (Step 3), $\eps^{-n} \Omega_K$ is a union of at most $(\delta^{-1}+1)$ disjoint unit cubes in $\R^d$ which are graphically connected (recall that a cube $Q_1$ neighbours to  a cube $Q_2$ if ${\rm dist}(Q_1,Q_2)=0$). Moreover, we know additionally that $0$ is the center of one of these sub-cubes.  Therefore,  $\eps^{-n} \Omega_K$ and $\eps^{-n} \wO_K$  belong to a finite collection of subsets of $\R^d$ which depends only on $d, \delta$. Thus we can use the second bound \eqref{eq:uncertainty-H-2} in Lemma \ref{lem:LUP-H} with the constant  $C_{\delta,\eps^{-n} \Omega_K,\eps^{-n} \wO_K}$ depending only on $d,s,\delta$, which leads immediately to \eqref{eq:H-final-2aa}. 

\medskip

 {\bf Case 2.} If $0\notin \Omega_K$, then using \eqref{eq:dist-0-Q} we have 
$$
|x|^{-2s} \le \frac{2^{2s}}{\eps^{2sn}}, \quad \forall x\in \Omega_K.
$$
Thus
\begin{align*}
	\Big\langle \Psi, \sum_{i=1}^N ((-\Delta_{x_i})^s_{|\wO} - \mathcal{C}_{s,d} |x_i|^{-2s}\1_{\Omega_K}(x_i) ) \Psi \Big\rangle \ge \Big\langle \Psi, \sum_{i=1}^N (-\Delta_{x_i})^s_{|\wO}  \Psi \Big\rangle - \mathcal{C}_{s,d} \frac{2^{2s}}{\eps^{2sn}} \int_{\Omega_K} \rho_\Psi
	\end{align*}
and hence \eqref{eq:H-final-2aa} follows from Lemma \ref{lem:LUP-II} (of course we have $C_{\rm GN} \ge C_{\rm HGN}$). 

Thus in both cases, \eqref{eq:H-final-2aa} always holds true. From \eqref{eq:H-final-2aa}, we use  
$$
\int_{\wO_K} \rho_\Psi \le 1+\delta, \quad \forall K\in G^{n,1}, \quad \forall n\ge 1
$$
and then sum  over all clusters. This gives 
 \begin{align} \label{eq:H-final-2}
&\sum_{n\ge 1} \sum_{K\subset G^{n,1}} \Big\langle \Psi, \sum_{i=1}^N \Big( (-\Delta_{x_i})^{s}_{|\wO_K} -\mathcal{C}_{s,d} |x_i|^{-2s} \1_{\Omega_K}(x_i) \Big) \Psi \Big\rangle \nn\\
&\ge \sum_{n\ge 1} \sum_{K\subset G^{n,1}} \Big( C_{\rm HGN} \frac{1-\delta}{(1+\delta)^{2s/d}} \int_{\Omega_K} \rho_\Psi^{1+2s/d}  - \frac{C_\delta}{\eps^{2sn}} \int_{\wO_K} \rho_\Psi \Big) 
	\end{align}

\bigskip
  \noindent 
 {\bf Step 4: Uncertainty principle III.}  Note that \eqref{eq:H-final-2} allows us to control the negative potential $-\mathcal{C}_{s,d} |x|^{-2s}$ on the clusters $K\subset G^{n,1}$, but we used  a bit more than the kinetic energy in $\Omega_K$ (we used $(-\Delta)^{s}_{|\wO_K}$).  Therefore, to complement for \eqref{eq:H-final-2} we have to deal with the  negative potential in the cubes $Q\in G^{n,0}$ using a bit less than the kinetic energy in $Q$. To be precise, we take a sub-cube $Q\in G^{n,0}$ and distinguish two cases. 
 
{\bf Case 1.} If $0\notin Q$, then we simply use  \eqref{eq:dist-0-Q} and deduce that
\begin{align} \label{eq:H-final-3a} 
\Big\langle \Psi, \sum_{i=1}^N ( - \mathcal{C}_{s,d} |x_i|^{-2s} \1_{Q}(x_i) ) \Psi \Big\rangle = - \mathcal{C}_{s,d} \int_Q \frac{\rho_\Psi}{|x|^{2s}} dx \ge  - \mathcal{C}_{s,d} \frac{2^{2s}}{|Q|^{2s/d}}  \int_Q \rho_\Psi. 
\end{align}

 {\bf Case 2.} If $0\in Q$, then $0$ is the center of $Q$. Denote  
 $$
 Q_0 := \frac{1}{4\sqrt{d}} Q = \left\{  \frac{1}{4\sqrt{d}}x \, |\, x\in Q\right\}. 
 $$
 Then using the first bound \eqref{eq:uncertainty-H-1} in Lemma \ref{lem:LUP-H}, we have
$$
\Big\langle \Psi, \sum_{i=1}^N ( (-\Delta_{x_i})^{s}_{|Q_0}- \mathcal{C}_{s,d} |x_i|^{-2s} \1_{Q_0}(x_i) ) \Psi \Big\rangle \ge  -  \frac{C}{|Q_0|^{2s/d}}  \int_{Q_0} \rho_\Psi. 
$$
 (Here we do not need the first term on the right side of \eqref{eq:uncertainty-H-1}.) Combining with 
 $$
 |x| \ge \frac{1}{8\sqrt{d}} |Q|^{1/d}, \quad \forall x \in Q\backslash Q_0
 $$
 we find that 
\begin{align} \label{eq:H-final-3bb} 
\Big\langle \Psi, \sum_{i=1}^N ( (-\Delta_{x_i})^{s}_{|Q_0}- \mathcal{C}_{s,d} |x_i|^{-2s} \1_{Q}(x_i) ) \Psi \Big\rangle \ge  -  \frac{C}{|Q|^{2s/d}}  \int_{Q} \rho_\Psi. 
\end{align}
Next, let us show that the set $Q_0$ is disjoint with $\wO_K$ for any cluster $K \subset G^{m,1}$ for any $m\ge 1$. Indeed, for any $Q'\in K\subset G^{m,1}$, using $0\notin Q'$ and $Q\cap Q'=\emptyset$ we obtain 
 $$
|x|  \ge \max\Big\{ \frac{1}{2}|Q'|^{1/d}, \frac{1}{2} |Q|^{1/d} \Big\} \ge \frac{1}{4} |Q'|^{1/d} + \frac{1}{4} |Q|^{1/d}, \quad \forall x\in Q'. 
 $$
 By the definition of the closure $\wO_K$ and the triangle inequality, we deduce that
  $$
|x|  \ge  \frac{1}{4} |Q|^{1/d}, \quad \forall x\in \wO_K.
 $$
On the other hand,
$$
|x| \le \frac{\sqrt{d}}{2} |Q_0|^{1/d} = \frac{1}{8} |Q|^{1/d},\quad \forall x\in Q_0. 
$$
Thus $Q_0$ is disjoint with $\wO_K$ for any cluster $K \subset G^{m,1}$ for any $m\ge 1$. Consequently, 
$$
(-\Delta)_{\R^d} \ge (-\Delta)^{s}_{|Q_0} + \sum_{m\ge 1} \sum_{K\subset G^{m,1}} (-\Delta)^{s}_{|\wO_K}. 
$$
Hence, from \eqref{eq:H-final-3bb} we conclude that if $0\in Q\in G^{n,0}$ for some $n\ge 1$, then 
\begin{align} \label{eq:H-final-3b} 
\Big\langle \Psi, \sum_{i=1}^N ( (-\Delta_{x_i})^{s}_{|\R^d}- \sum_{m\ge 1} \sum_{K\subset G^{m,1}} (-\Delta_{x_i})^{s}_{|\wO_K} - \mathcal{C}_{s,d} |x_i|^{-2s} \1_{Q}(x_i) ) \Psi \Big\rangle \ge  -  \frac{C}{|Q|^{2s/d}}  \int_{Q} \rho_\Psi. 
\end{align}

{\bf Combining \eqref{eq:H-final-3a} and \eqref{eq:H-final-3b}:} Since the sub-cubes in $\bigcup_{n\ge 1} G^{n,0}$ are disjoint, there is at most one sub-cube containing $0$. Therefore, by summing over all $Q\in G^{n,0}$, $n\ge 1$ we obtain from   \eqref{eq:H-final-3a} and \eqref{eq:H-final-3b} that  
\begin{align} \label{eq:H-final-3} 
&\Big\langle \Psi, \sum_{i=1}^N \Big( (-\Delta_{x_i})^{s}_{|\R^d}- \sum_{m\ge 1} \sum_{K\subset G^{m,1}} (-\Delta_{x_i})^{s}_{|\wO_K} \Big) \Psi \Big\rangle  - \sum_{n\ge 1} \sum_{Q\in G^{n,0}} \Big\langle \Psi, \sum_{i=1}^N \mathcal{C}_{s,d} |x_i|^{-2s} \1_{Q}(x_i)  \Psi \Big\rangle \nn\\
&\ge  - \sum_{n\ge 1} \sum_{Q\in G^{n,0}} \frac{C}{|Q|^{2s/d}}  \int_{Q} \rho_\Psi. 
\end{align}

\bigskip

\noindent
{\bf Step 5: Conclusion of the lower bound.}  Summing \eqref{eq:H-final-2} and \eqref{eq:H-final-3} we get 
\begin{align} \label{eq:H-final-4}
	 \Big\langle \Psi, \sum_{i=1}^N \Big( (-\Delta_{x_i})^{s}_{|\R^d} - \mathcal{C}_{s,d}|x_i|^{-2s} \Big) \Psi \Big\rangle 	&\ge C_{\rm HGN} \frac{1-\delta}{(1+\delta)^{2s/d}}  \sum_{n\ge 1} \sum_{K\subset G^{n,1}} \int_{\Omega_K} \rho_\Psi^{1+2s/d} \nn\\
	 &\quad  -  \sum_{n\ge 1} \frac{C_\delta}{\eps^{2sn}} \Big(\sum_{Q\in G^{n,0}} \int_Q \rho_\Psi + \sum_{K \subset G^{n,1}} \int_{\wO_K} \rho_\Psi  \Big).
\end{align}
Moreover, recall the exclusion bound in \eqref{eq:final-3}:
$$
 \Big\langle \Psi, \sum_{1\le i<j\le N} \frac{\lambda}{|x_i-x_j|^{2s} }  \Psi\Big\rangle \ge  \sum_{n\ge 1} \frac{\lambda}{C_\delta \eps^{2sn}} \Big( \sum_{Q\in G^{n,0}} \int_Q \rho_\Psi +  \sum_{K \subset G^{n,1}} \int_{\wO_K} \rho_\Psi \Big)
$$
Finally, we multiply \eqref{eq:H-final-1} with $\delta^{s/d}$, multiply \eqref{eq:H-final-4} with $(1-\delta^{s/d})$, and then sum them with the above exclusion bound. For any given $\delta\in (0,1)$, we can choose $\lambda>0$ sufficiently large such that 
$$
\frac{\lambda}{C_\delta} \ge C_\delta
$$ 
namely the interaction energy from the exclusion bound dominates the error terms in \eqref{eq:H-final-1}  and \eqref{eq:H-final-4}. We thus obtain  
\begin{align} \label{eq:H-final-5}
	 &\Big\langle \Psi, \sum_{i=1}^N \Big( (-\Delta_{x_i})^{s}_{|\R^d} - \mathcal{C}_{s,d}|x_i|^{-2s} \Big) \Psi \Big\rangle 	\nn\\
	 &\ge \frac{1}{C\delta^{s/d}}\sum_{n\ge 1} \sum_{Q\subset G^{n,0}} \int_{Q} \rho_\Psi^{1+2s/d} + C_{\rm HGN} \frac{(1-\delta)(1-\delta^{s/d})}{(1+\delta)^{2s/d}} \sum_{n\ge 1} \sum_{K\subset G^{n,1}} \int_{\Omega_K} \rho_\Psi^{1+2s/d}  \nn\\
	 &\ge  \min\Big\{ \frac{1}{C\delta^{s/d}}, C_{\rm HGN} \frac{(1-\delta)(1-\delta^{s/d})}{(1+\delta)^{2s/d}}\Big\}  \sum_{n\ge 1} \Big( \sum_{Q\subset G^{n,0}} \int_{Q} \rho_\Psi^{1+2s/d} +  \sum_{K\subset G^{n,1}} \int_{\Omega_K} \rho_\Psi^{1+2s/d} \Big) \nn\\
	 &=  \min\Big\{ \frac{1}{C\delta^{s/d}}, C_{\rm HGN} \frac{(1-\delta)(1-\delta^{s/d})}{(1+\delta)^{2s/d}}\Big\} \int_{\R^d} \rho_\Psi^{1+2s/d} .
	 \end{align}
Here in the last equality we have used the covering property \eqref{eq:inclusion-1}. Thus we have proved that for any constant $\delta\in (0,1)$ and for $\lambda>0$ sufficiently large,  the optimal constant in the Hardy--Lieb--Thirring inequality satisfies
$$
C_{\rm HLT} (\lambda) \ge \min\Big\{ \frac{1}{C\delta^{s/d}}, C_{\rm HGN} \frac{(1-\delta)(1-\delta^{s/d})}{(1+\delta)^{2s/d}}\Big\}. 
$$
By taking $\lambda\to \infty$ and  then $\delta\to 0$, we get 
$$
\liminf_{\lambda\to \infty} C_{\rm HLT} (\lambda) \ge C_{\rm HGN} .
$$

\bigskip
\noindent
{\bf Step 6: Upper bound.} To conclude, let us prove that 
$$C_{\rm HLT} (\lambda) \le C_{\rm HGN}, \quad \lambda>0.$$
We construct a 2-body state as follows. Take $u,v\in C_c^\infty(\R^d)$ with $\|u\|_{L^2(\R^d)}=\|v\|_{L^2(\R^d)}=1$. Then for any $z\in \R^d$ we consider the trial state 
$$
\Psi_z(x,y) =u(x) v(y-z). 
$$ 
It is straightforward to see that 
\begin{align*}
C_{\rm HLT} (\lambda) &\le \lim_{|z| \to \infty} \frac{ \Big\langle \Psi_z, \sum_{i=1}^2 \Big( (-\Delta_{x_i})^{s}_{|\R^d} - \mathcal{C}_{s,d}|x_i|^{-2s} \Big) \Psi_z \Big\rangle}{\int_{\R^d}\rho_\Psi^{1+2s/d}} \\
&= \frac{\Big\langle u,  \Big( (-\Delta)^{s} - \mathcal{C}_{s,d}|x|^{-2s} \Big) u \Big\rangle + \Big\langle v, (-\Delta)^{s}   v \Big\rangle   }{\int_{\R^d} |u|^{2(1+2s/d)}+ \int_{\R^d}|v|^{2(1+2s/d)}  }.
\end{align*}
Then we replace $v$ by $v_\ell(x)=\ell^{d/2}v(\ell x)$. Note that
$$
\Big\langle v_\ell, (-\Delta)^{s}   v_\ell \Big\rangle = \ell^{2s} \Big\langle v, (-\Delta)^{s}  v \Big\rangle, \quad  \int_{\R^d}|v_\ell|^{2(1+2s/d)}  = \ell^{2s} \int_{\R^d}|v_\ell|^{2(1+2s/d)} .
$$
Therefore, by taking $\ell\to 0^+$ we obtain
\begin{align*}
C_{\rm HLT} (\lambda) &\le \lim_{\ell \to 0^+}  \frac{\Big\langle u,  \Big( (-\Delta)^{s} - \mathcal{C}_{s,d}|x|^{-2s} \Big) u \Big\rangle + \Big\langle v_\ell , (-\Delta)^{s}   v_\ell \Big\rangle   }{\int_{\R^d} |u|^{2(1+2s/d)}+ \int_{\R^d}|v_\ell|^{2(1+2s/d)}  } \\
&= \frac{\Big\langle u,  \Big( (-\Delta)^{s} - \mathcal{C}_{s,d}|x|^{-2s} \Big) u \Big\rangle   }{\int_{\R^d} |u|^{2(1+2s/d)}  }, \quad \forall u\in C_c^\infty(\R^d).
\end{align*}
Optimizing over $u$ we find that $C_{\rm HLT} (\lambda) \le C_{\rm HGN}$ for all $\lambda>0$. 

Thus we conclude that  $C_{\rm HLT} (\lambda) \to  C_{\rm HGN}$ as $\lambda\to \infty$. This completes the proof of Theorem \ref{thm-main2}. 
\end{proof}

\end{document}